\documentclass[11pt,letterpaper]{article}
\usepackage[margin=1in]{geometry}
\usepackage{lmodern}
\usepackage[T1]{fontenc}
\usepackage[utf8]{inputenc}
\usepackage[english]{babel}
\usepackage{tikz,amsmath,amssymb,amsfonts,latexsym,graphicx,amsthm,algorithmicx,comment}
\usepackage{bbm}
\usepackage{caption}
\usepackage{subcaption}
\usepackage{afterpage}

\usepackage{thmtools}
\declaretheorem{theorem}

\usepackage{algorithm}
\usepackage[noend]{algpseudocode}
\usepackage{hyperref}
\usepackage[capitalise]{cleveref}

\theoremstyle{plain}
\newtheorem{lemma}[theorem]{Lemma}

\newtheorem{fact}[theorem]{Fact}
\newtheorem{corollary}[theorem]{Corollary}
\newtheorem{remark}[theorem]{Remark}
\theoremstyle{definition}

\newtheorem{definition}[theorem]{Definition}

\newcommand{\OPT}{\mathrm{OPT}}
\newcommand{\TSP}{\mathrm{TSP}}

\newcommand{\opt}{\mathrm{opt}}
\newcommand{\spine}{\mathrm{spine}}

\newcommand{\dist}{\mathrm{dist}}
\newcommand{\cost}{\mathrm{cost}}
\newcommand{\DP}{\mathrm{DP}}
\newcommand{\eps}{\epsilon}
\newcommand{\demand}{\mathrm{demand}}

\title{A Tight $(1.5+\epsilon)$-Approximation for \\Unsplittable Capacitated Vehicle Routing on Trees}
\date{}
\author{Claire Mathieu\footnote{CNRS, Paris, France, e-mail: \texttt{claire.mathieu@irif.fr}.}  \and Hang Zhou\footnote{Ecole Polytechnique, Institut Polytechnique de Paris, France, e-mail: \texttt{hzhou@lix.polytechnique.fr}.}}

\begin{document}

\maketitle

\begin{abstract}
In the unsplittable capacitated vehicle routing problem (UCVRP) on trees, we are given a rooted tree with edge weights and a subset of vertices of the tree called terminals. Each terminal is associated with a positive demand between 0 and 1. The goal is to find a minimum length collection of tours starting and ending at the root of the tree such that the demand of each terminal is covered by a single tour (i.e., the demand cannot be split), and the total demand of the terminals in each tour does not exceed the capacity of 1.

For the special case when all terminals have equal demands, a long line of research culminated in a quasi-polynomial time approximation scheme [Jayaprakash and Salavatipour, SODA 2022] and a polynomial time approximation scheme [Mathieu and Zhou, ICALP 2022].

In this work, we study the general case when the terminals have arbitrary demands. Our main contribution is a polynomial time $(1.5+\epsilon)$-approximation algorithm for the UCVRP on trees. This is the first improvement upon the 2-approximation algorithm more than 30 years ago [Labbé, Laporte, and Mercure, Operations Research, 1991]. Our approximation ratio is essentially best possible, since it is NP-hard to approximate the UCVRP on trees to better than a 1.5 factor.

\end{abstract}

\thispagestyle{empty}
\setcounter{page}{0}

\newpage
\pagenumbering{arabic}

\setcounter{page}{1}

\section{Introduction}
In the \emph{unsplittable capacitated vehicle routing problem (UCVRP)} on \emph{trees}, we are given a rooted tree with edge weights and a subset of vertices of the tree called \emph{terminals}.
Each terminal is associated with a positive \emph{demand} between 0 and 1.
The root of the tree is called the \emph{depot}.
The goal is to find a minimum length collection of tours starting and ending at the depot such that the demand of each terminal is covered by a \emph{single} tour (i.e., the demand cannot be split), and the total demand of the terminals in each tour does not exceed the \emph{capacity} of 1.

The UCVRP on trees has been well studied in the special setting when all terminals have \emph{equal} demands:\footnote{Up to scaling, the equal demand setting is equivalent to the \emph{unit demand} version of the capacitated vehicle routing problem in which each terminal has unit demand, and the capacity of each tour is a positive integer $k$.}
Hamaguchi and Katoh~\cite{hamaguchi1998capacitated} gave a polynomial time 1.5-approximation; the approximation ratio was improved to $1.35078$ by Asano, Katoh, and Kawashima~\cite{asano2001new} and was further reduced to $4/3$ by Becker~\cite{becker2018tight}; Becker and Paul~\cite{becker2019framework} gave a bicriteria polynomial time approximation scheme; and very recently, Jayaprakash and Salavatipour~\cite{jayaprakash2021approximation} gave a quasi-polynomial time approximation scheme, based on which Mathieu and Zhou~\cite{MZ22} designed a polynomial time approximation scheme.

In this work, we study the UCVRP on trees in the general setting when the terminals have \emph{arbitrary} demands.
Our main contribution is a polynomial time $(1.5+\eps)$-approximation algorithm (\cref{thm:main}).
This is the first improvement upon the 2-approximation algorithm of Labbé, Laporte, and Mercure~\cite{labbe1991capacitated} more than 30 years ago.
Our approximation ratio is essentially best possible, since it is NP-hard to approximate the UCVRP on trees to better than a 1.5 factor~\cite{golden1981capacitated}.

\begin{theorem}
\label{thm:main}
For any $\eps>0$, there is a polynomial time $(1.5+\eps)$-approximation algorithm for the unsplittable capacitated vehicle routing problem on trees.
\end{theorem}

The UCVRP on trees generalizes the UCVRP on \emph{paths}.
The latter problem has been studied extensively due to its applications in scheduling, see \cref{sec:related}.
Previously, the best approximation ratio for the UCVRP on paths was 1.6 due to Wu and Lu~\cite{wu2020capacitated}.
As an immediate corollary of \cref{thm:main}, we obtain a polynomial time $(1.5+\eps)$-approximation algorithm for the UCVRP on paths.
This ratio is essentially best possible, since it is NP-hard to approximate the UCVRP on paths to better than a 1.5 factor (\cref{sec:hardness-path}).

\subsection{Related Work}
\label{sec:related}


\paragraph{UCVRP on paths.}
The UCVRP on paths is equivalent to the scheduling problem of minimizing the makespan on a single batch processing machine  with non-identical job sizes~\cite{uzsoy1994scheduling}.
Many heuristics have been proposed and evaluated empirically, e.g., \cite{uzsoy1994scheduling,dupont2002minimizing,melouk2004minimizing,damodaran2006minimizing,kashan2006effective,parsa2010branch,chen2011scheduling,al2015constrained,muter2020exact}.

The UCVRP on paths has also been studied in special cases.
For example, in the special case when the optimal value is at least $\Omega(1/\eps^6)$ times the maximum distance between any terminal and the depot, asymptotic polynomial time
approximation schemes are known~\cite{das2010train,rothvoss2012entropy,chen2013train}.\footnote{The UCVRP on paths is called the \emph{train delivery problem} in~\cite{das2010train,rothvoss2012entropy,chen2013train}.}
In contrast, the algorithm in \cref{thm:main} applies to any path instance (and more generally any tree instance).

\paragraph{UCVRP on general metrics.}
The first constant-factor approximation algorithm for the UCVRP on general metrics is due to Altinkemer and  Gavish~\cite{altinkemer1987heuristics}.
The approximation ratio was only recently improved in work by Blauth, Traub, and Vygen~\cite{blauth2021improving}, and then further  by Friggstad, Mousavi, Rahgoshay, and Salavatipour~\cite{friggstad2021improved}, so that the best-to-date  approximation ratio  stands at roughly $3.194$~\cite{friggstad2021improved}.

\section{Overview of Techniques}\label{sec:overview}


At a high level, our approach is to modify the problem and add enough structural constraints so that the structured  problem contains a $(1.5+O(\epsilon))$-approximate solution and can be solved in polynomial time by dynamic programming.

\subsection{Preprocessing}
We start by some preprocessing as in~\cite{MZ22}.
We assume without loss of generality that every vertex in the tree has two children, and the terminals are the leaf vertices of the tree~\cite{MZ22}.
Furthermore, we assume that the tree has bounded distances (\cref{sec:bounded-dist}).
Next, we decompose the tree into \emph{components} (\cref{fig:tree} and \cref{sec:components}).

\begin{figure}[h]
    \centering
    \includegraphics[scale=0.37]{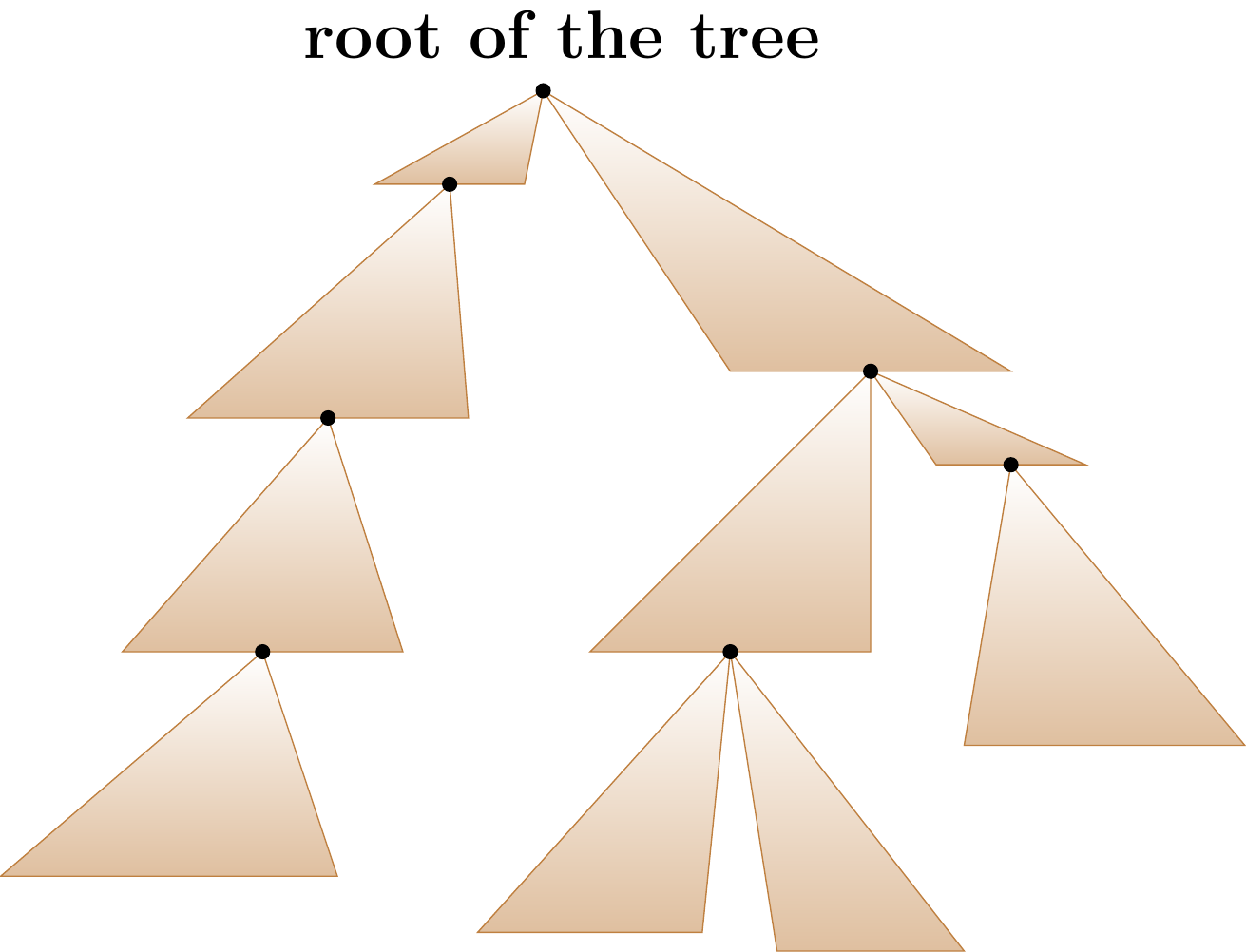}
    \caption{\small Decomposition of the tree into \emph{components}. Figure extracted from \cite{MZ22}. Each brown triangle represents a \emph{component}. Each component has a \emph{root} vertex and at most one \emph{exit} vertex. }
    \label{fig:tree}
\end{figure}

\begin{figure}[t]
\centering
    \begin{subfigure}{1\textwidth}
    \centering
    \includegraphics[scale=0.37]{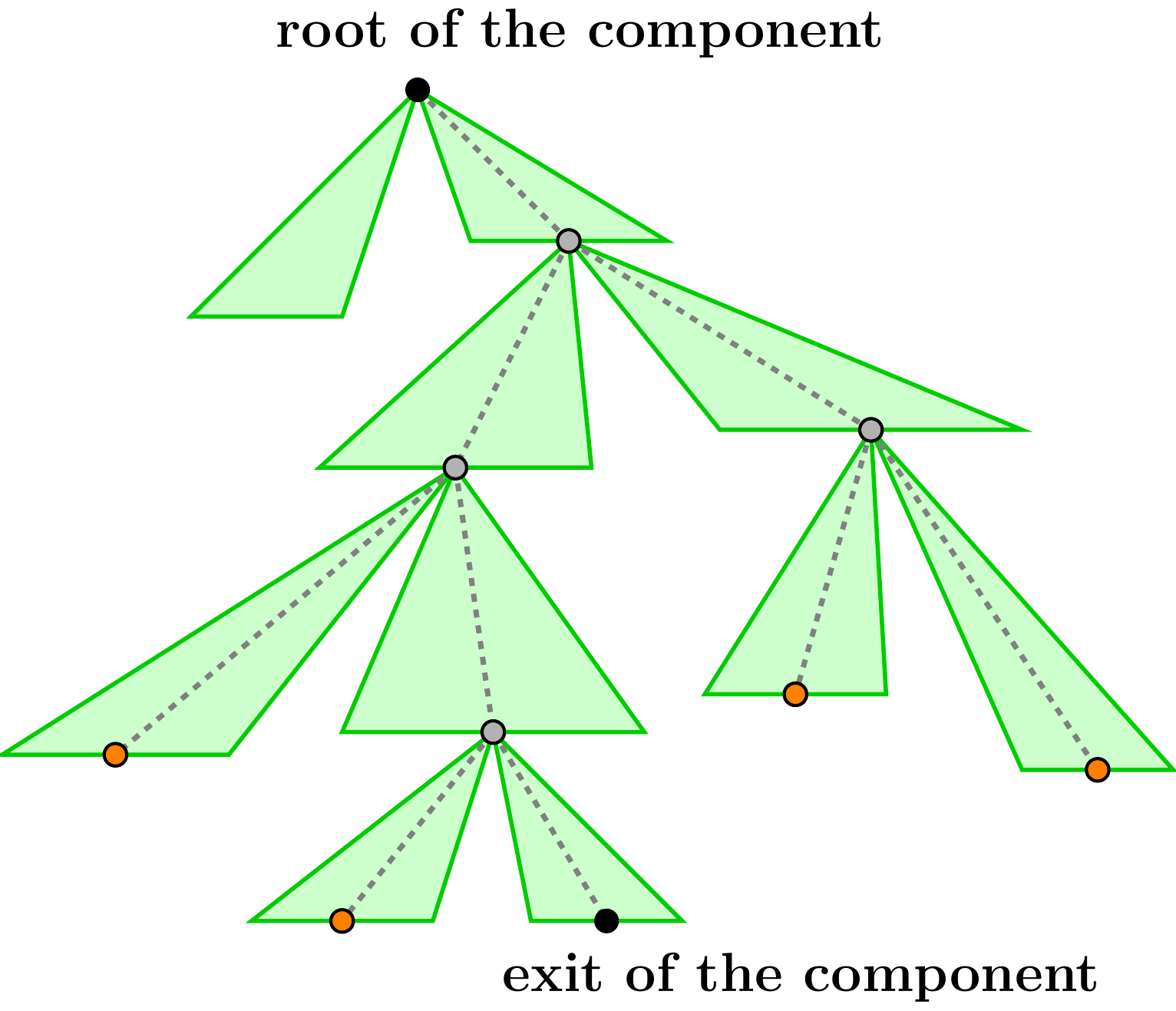}
    \caption{Decomposition of a component into \emph{blocks}.
    The orange nodes represent the \emph{big} terminals in the component.
    The black nodes represent the \emph{root} and the \emph{exit} vertices of the component (defined in \cref{lem:decomposition}).
    The gray nodes are the branching vertices in the subtree spanning the orange and the black nodes.
    Splitting the component at the orange, the black, and the gray nodes results in a set of \emph{blocks}, represented by green triangles.
    Each block has a \emph{root} vertex and at most one \emph{exit} vertex.
    See \cref{sec:decomposition-level-1}.\vspace{8mm}}
    \label{fig:component-block}
    \end{subfigure}
    \begin{subfigure}{\textwidth}
    \centering
    \includegraphics[scale=0.37]{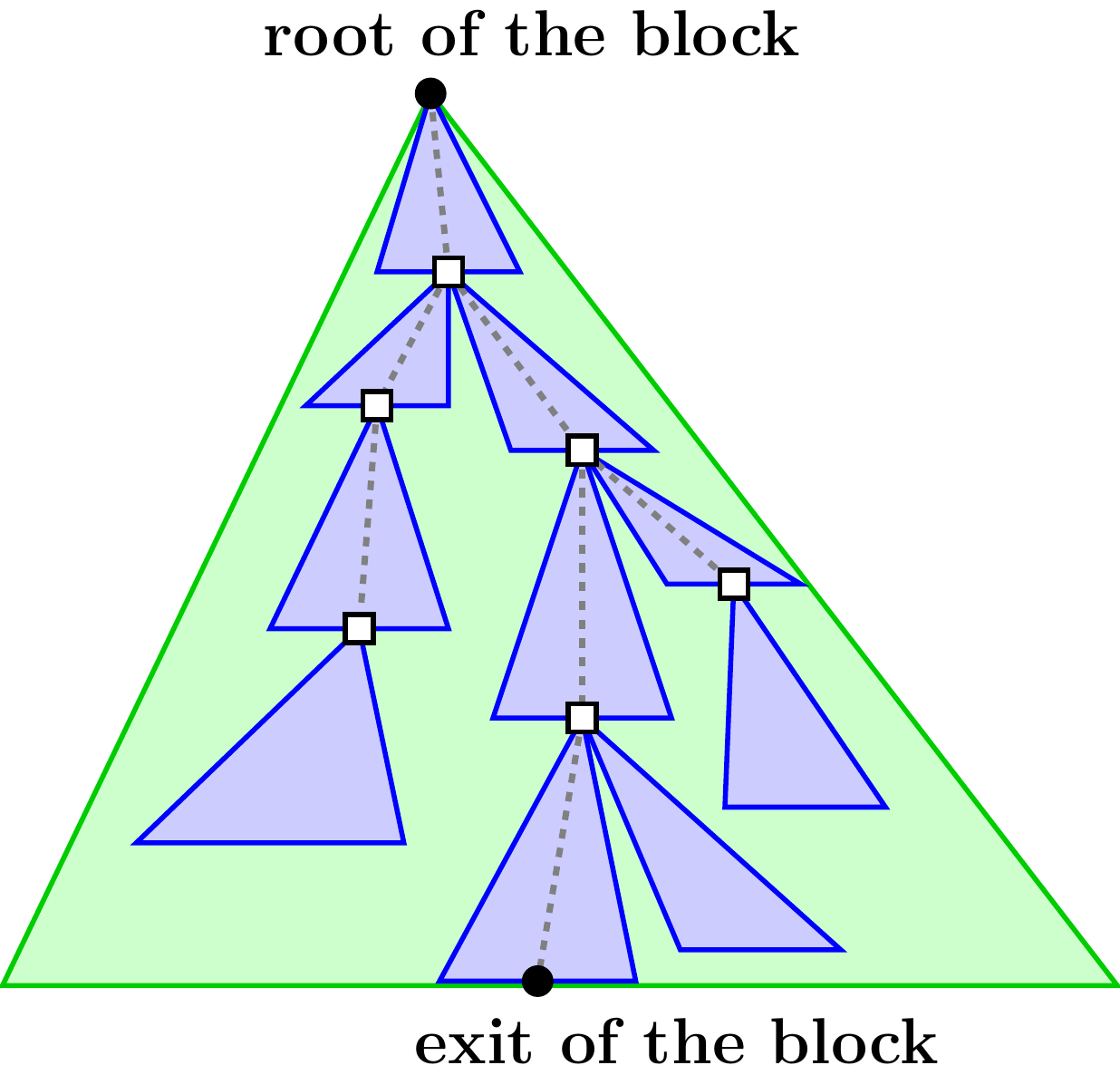}
    \caption{Decomposition of a block into \emph{clusters}.
    The green triangle represents a block.
    Each blue triangle represents a \emph{cluster}.
    Each cluster has a \emph{root} vertex and at most one \emph{exit} vertex.
    A cluster is \emph{passing} if it has an exit vertex, and is \emph{ending} otherwise.
    Each passing cluster has a \emph{spine} (dashed).
    See \cref{sec:decomposition-level-2}.\vspace{8mm}}
    \label{fig:block-cluster}
    \end{subfigure}
    \hfill
    \begin{subfigure}{\textwidth}
    \centering
    \includegraphics[scale=0.37]{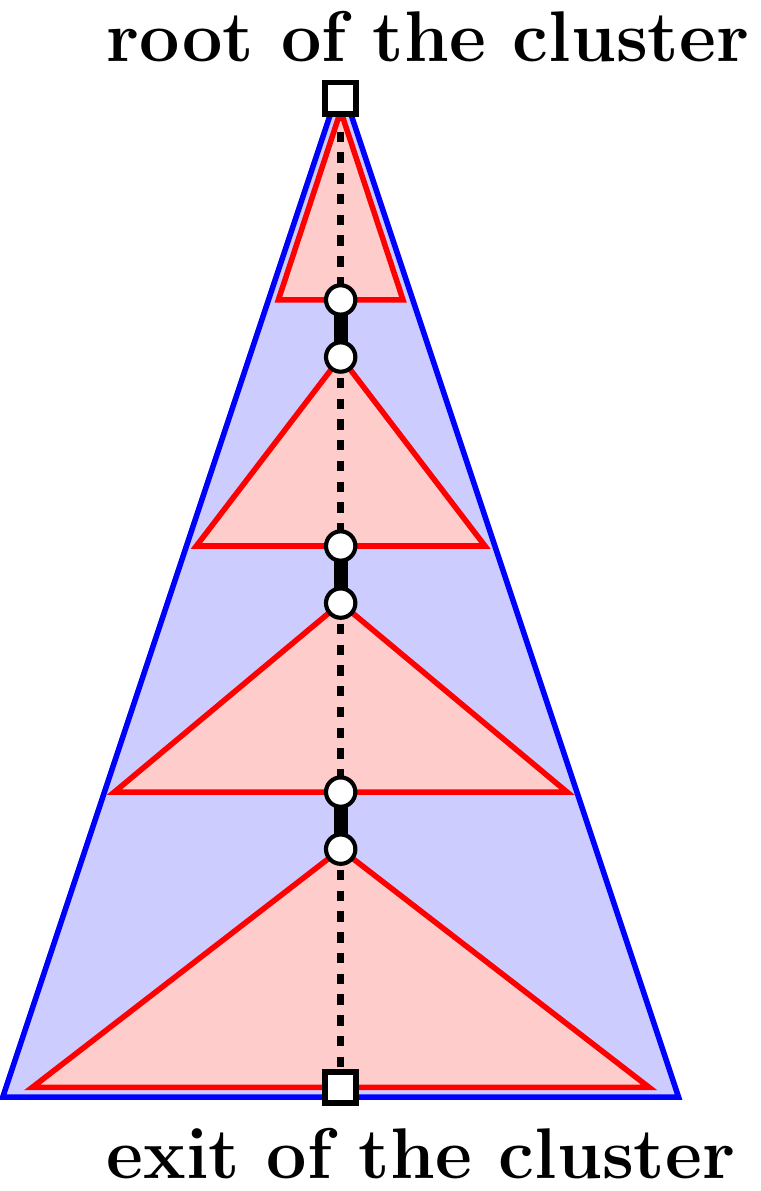}
    \caption{Decomposition of a passing cluster into \emph{cells}.
    The blue triangle represents a passing cluster.
    Removing the thick edges from the cluster results in a set of at most $1/\eps$ \emph{cells}.
    Each red triangle represents a cell. Each of those cells has a \emph{root} vertex, an \emph{exit} vertex, and a \emph{spine} (dashed).
    See \cref{sec:decomposition-level-3}.
    }
    \label{fig:cluster-cell}
    \end{subfigure}
    \caption{Three-level decomposition of a component.\vspace{5mm}}
    \label{fig:three-level}
\end{figure}

\subsection{Solutions Within Each Component}
\label{sec:outline-decomposition}

A significant difficulty is to compute solutions within each component.
It would be natural to attempt to extend the approach in the setting when all terminals have equal demands~\cite{MZ22}. In that setting,  the \emph{demands of the subtours}\footnote{The \emph{demand of a subtour} is the total demand of the terminals visited by that subtour.} in each component are among a polynomial number of values; since the component is visited by a constant  number of tours in a near-optimal solution, that solution inside the component can be computed exactly in polynomial time using a simple dynamic program.
However, when the terminals have arbitrary demands, the demands of the subtours in each component might be among an exponential number of values.\footnote{For example, consider a component that is a star graph with $\Theta(n)$ leaves, where the $i^{\rm th}$ leaf has demand $1/2^i$.}
Indeed, unless $P=NP$, we cannot compute in polynomial time a better-than-$1.5$ approximate solution inside a component, since that problem generalizes the bin packing problem (\cref{sec:hardness-tree-component}).

To compute in polynomial time good approximate solutions within each component, at a high level, we simplify the solution structure in each component, so that the demands of the subtours in that component are among a \emph{constant} $O_\eps(1)$ number of values, while increasing the cost of the solution by at most a multiplicative factor $1.5+O(\eps)$.

Where does the $1.5$ factor come from? Intuitively, our construction creates an additional subtour to cover a selected subset of terminals, charging each edge on that  subtour to \emph{two} existing subtours using that edge, thus adding a $0.5$ factor to the cost.

In the rest of this section, we explain our approach in more details.

\subsubsection{Multi-Level Decomposition (\cref{sec:multileveldecomposition})}
First, we distinguish \emph{big} and \emph{small} terminals depending on their demands.
The number of big terminals in a component is $O_\eps(1)$.
Next, we partition the small terminals of a component into $O_\eps(1)$ parts using a \emph{multi-level decomposition} as follows.
In the first level, a component is decomposed into $O_\eps(1)$ \emph{blocks} so that all terminals strictly inside a block are small; see \cref{fig:component-block,sec:decomposition-level-1}.
In the second level, each block is decomposed into $O_\eps(1)$ \emph{clusters} so that the overall demand of each cluster is roughly an $\eps$ fraction of the demand of a component; see \cref{fig:block-cluster,sec:decomposition-level-2}.
Intuitively, the clusters are such that, \emph{if} we assign the small terminals of each cluster to a single subtour, the subtour capacities would be violated only slightly.
We define the \emph{spine} of a cluster to be the path traversing that cluster.
In the third level, each cluster is decomposed into $O_\eps(1)$ \emph{cells} so that the spine of each cell is roughly an $\eps$ fraction of the spine of a cluster; see \cref{fig:cluster-cell,sec:decomposition-level-3}.

\begin{figure}[h]
    \centering
    \includegraphics[scale=0.25]{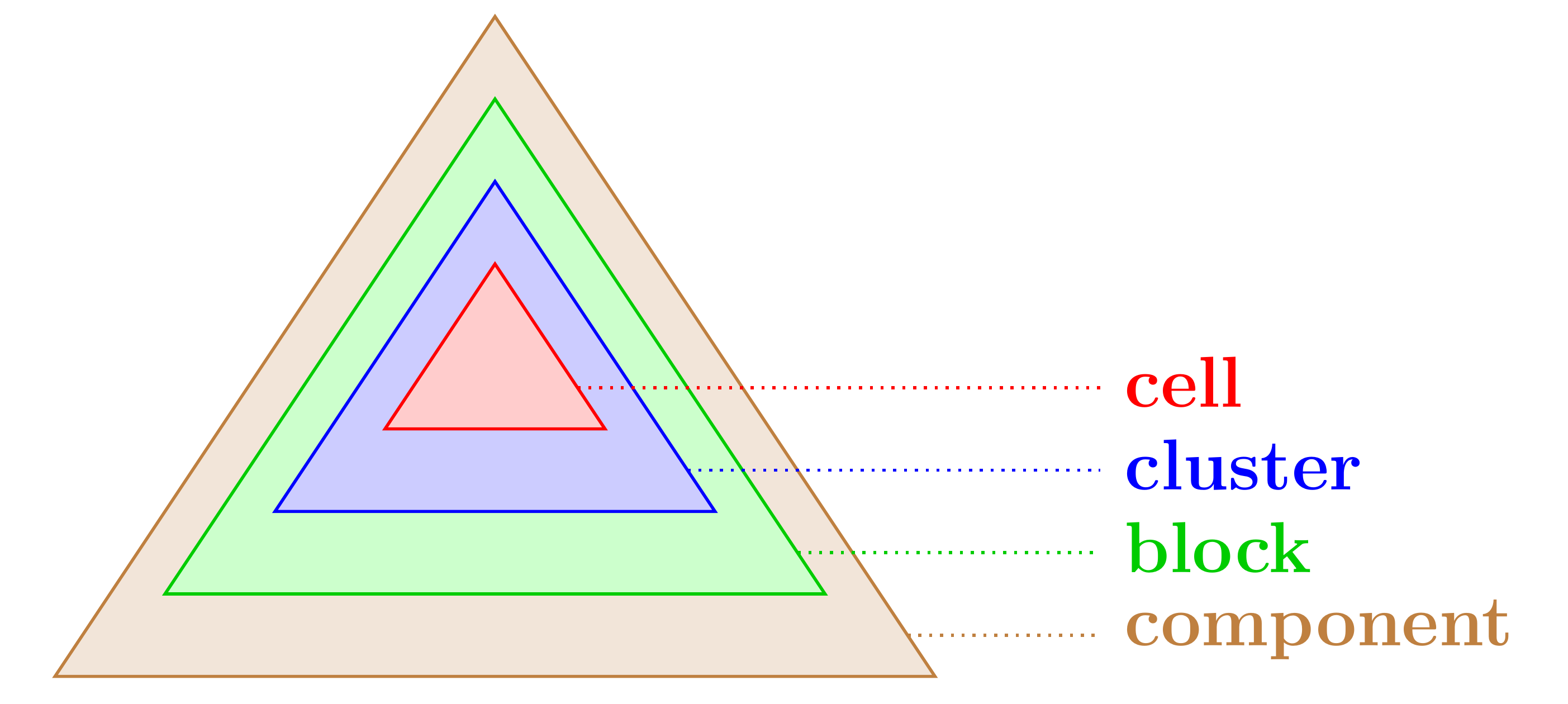}
    \caption{Relation of the multiple levels in the decomposition.}
    \label{fig:nesting}
\end{figure}



\subsubsection{Simplifying the Local Solution (\cref{sec:local})}
\label{sec:Overview-simplify}
The main technical contribution in this paper is the Local Theorem (\cref{thm:local}), which simplifies a local solution inside a component so that, \emph{in each cell, a single subtour visits all small terminals}, while increasing the cost of the local solution by at most a multiplicative factor $1.5+O(\eps)$.
The Local Theorem builds upon techniques from \cite{MZ22} together with substantial new ideas.

A first attempt is to combine all subtours of a cluster into a single subtour.
However, there are two obstacles.
First, the resulting subtours in the component are no longer connected; reconnecting those subtours would require including the spines of the clusters, which would be too expensive.
Secondly, the resulting subtours in the component may exceed their capacities, so an extra cost is needed to reduce the demands of those subtours; in the equal demand setting~\cite{MZ22}, that extra cost is at most an $\eps$ fraction of the solution cost, but this is no longer the case in the arbitrary demand setting.

We overcome those obstacles thanks to the decomposition of a cluster into \emph{cells}.
In the analysis, we introduce the technical concept of \emph{threshold cells} (\cref{fig:threshold_before}), and we ensure that each cluster contains \emph{at most one} threshold cell.
A crucial step is to \emph{include the spine subtour of the threshold cell} into the solution (\cref{fig:threshold_after}).
This enables us to reassign all small terminals of each cell to a single subtour, without losing connectivity, while only slightly violating the subtour capacities.

To reduce the demand of each subtour exceeding capacity, we select some cells from that subtour, and we remove all pieces in that subtour belonging to those cells.
We show that each removed piece is connected to the root through at least \emph{two} subtours in the solution (\cref{lem:connected}, see \cref{fig:analysis}).
That property is a main technical novelty in this paper.
It enables us to reconnect all removed pieces with an extra cost of at most \emph{half} of the solution cost (\cref{lem:extra-cost-2}), hence an approximation ratio of $1.5+O(\eps)$.

\subsection{Postprocessing}
As in \cite{MZ22}, we modify the tree so that it has only $O_\eps(1)$ levels of components; see \cref{sec:hat-T}.
Consider a near-optimal solution in the resulting tree (\cref{thm:MZ22-opt}).
After applying the Local Theorem (\cref{thm:local}) to simplify the local solutions in all components, we obtain a global solution (\cref{thm:global}); see \cref{sec:global}.
We observe that the possible subtour demands in that global solution are within a polynomial number of values (\cref{fact:Y}).
Furthermore, we leverage the  \emph{adaptive rounding} technique due to Jayaprakash and Salavatipour~\cite{jayaprakash2021approximation} and we show that, in each subtree, the subtour demands are among a constant $O_\eps(1)$ number of values (\cref{thm:structure}); see \cref{sec:structure}.

Finally, we design a polynomial time dynamic program to compute the best solution that satisfies the structural constraints established previously.
The computed solution is a $(1.5+O(\eps))$-approximation.
See \cref{sec:DP}.

\begin{remark}
For more general metrics, such as graphs of bounded treewidth and the Euclidean space, it is challenging to simplify the solution structure due to the lack of the \emph{unique} spine in a subproblem.
It is an open question to design optimal approximation algorithms in those metrics.
\end{remark}

\section{Preliminaries}
\label{sec:notations}
Let $T$ be a rooted tree $(V,E)$ with edge weights $w(u,v)\geq 0$ for all $(u,v)\in E$.
Let $n$ denote the number of vertices in $V$.
The \emph{cost} of a tour (resp.\ a subtour) $t$, denoted by $\cost(t)$, is the overall weight of the edges on $t$.
For a set $S$ of tours (resp.\ subtours), the \emph{cost} of $S$, denoted by $\cost(S)$, is $\sum_{t\in S}\cost(t)$.

\begin{definition}[UCVRP on trees]
An instance of the \emph{unsplittable capacitated vehicle routing problem (UCVRP)} on \emph{trees} consists of
\vspace{-0.5em}
\begin{itemize}
\itemsep-0.2em
\item an edge weighted \emph{tree} $T=(V,E)$ with \emph{root} $r\in V$ representing the \emph{depot},
    \item a set $V'\subseteq V$ of \emph{terminals},
    \item for each terminal $v\in V'$, a \emph{demand} of $v$, denoted by $\demand(v)$, which belongs to $(0,1]$.
\end{itemize}
\noindent A feasible solution is \emph{a set of tours} such that
\vspace{-0.5em}
\begin{itemize}
\itemsep-0.2em
\item each tour starts and ends at $r$,
    \item the demand of each terminal is covered by a \emph{single} tour,
    \item the total demand of the terminals covered by each tour does not exceed the \emph{capacity} of 1.
\end{itemize}
The goal is to find a feasible solution of minimum cost.
\end{definition}

For any two vertices $u,v\in V$, let $\dist(u,v)$ denote the distance between $u$ and $v$ in the tree $T$.

We say that a tour (resp.\ a subtour) \emph{visits} a terminal if it covers the demand of that terminal.
For technical reasons, we allow \emph{dummy} terminals of appropriate demands to be included artificially.
The \emph{demand} of a tour (resp.\ a subtour) $t$, denoted by $\demand(t)$, is defined to be the total demand of all terminals (including dummy terminals) visited by $t$.

\subsection{Reduction to Bounded Distance Instances}
\label{sec:bounded-dist}
\begin{definition}[Definition 3 in \cite{MZ22}]
\label{def:bounded-distance}
Let $D_{\min}$ (resp.\ $D_{\max}$) denote the minimum (resp.\ maximum) distance between the depot and any terminal in the tree $T$.
We say that $T$ has \emph{bounded distances} if
$D_{\max}<(1/\eps)^{(1/\eps)-1}\cdot D_{\min}.$
\end{definition}

The next theorem (\cref{thm:dist-reduce}) enables us to assume without loss of generality that the tree $T$ has bounded distances.

\begin{theorem}[Theorem~5 in \cite{MZ22}]
\label{thm:dist-reduce}
For any $\rho\geq 1$, if there is a polynomial time $\rho$-approximation algorithm for the UCVRP on trees with \emph{bounded distances}, then there is a polynomial time $(1+5\eps) \rho$-approximation algorithm for the UCVRP on trees with \emph{general distances}.
\end{theorem}


\subsection{Decomposition Into Components (\cref{fig:tree})}

\label{sec:components}

We decompose the tree $T$ into components as in \cite{MZ22}.

\begin{lemma}[Lemma 9 in~\cite{MZ22}]
\label{lem:decomposition}
Let $\Gamma=12/\eps$.
There is a polynomial time algorithm to compute a partition of the edges of the tree $T$ into a set $\mathcal{C}$ of \emph{components}, such that all of the following properties are satisfied:
\vspace{-0.5em}
\begin{enumerate}
\itemsep-0.2em
    \item Every component $c\in \mathcal{C}$ is a connected subgraph of $T$;
    the \emph{root} vertex of the component $c$, denoted by $r_c$, is the vertex in $c$ that is closest to the depot.
    \item A component $c$ shares vertices with other components at vertex $r_c$ and possibly at one other vertex, called the \emph{exit} vertex of the component $c$ and denoted by $e_c$.
    We say that $c$ is an \emph{internal} component if $c$ has an exit vertex, and is a \emph{leaf} component otherwise.
    \item The total demand of the terminals in each component $c\in \mathcal{C}$ is at most $2\Gamma$.
    \item The number of components in  $\mathcal{C}$ is at most $\max\{1,3\cdot \demand(T)/\Gamma\}$, where $\demand(T)$ denotes the total demand of the terminals in the tree $T$.
\end{enumerate}
\end{lemma}


\begin{definition}[\cite{MZ22}]
\label{def:subtour-component}
Let $c\in\mathcal{C}$ be any component.
A \emph{subtour in component~$c$} is a path $t$ that starts and ends at the root $r_c$ of component $c$, and such that every vertex on $t$ is in component $c$.
We say that $t$ is a \emph{passing subtour} if $c$ has an exit vertex and that vertex belongs to $t$, and is an \emph{ending subtour} otherwise.
\end{definition}

\section{Multi-Level Decomposition in a Component}
\label{sec:multileveldecomposition}

Let $c\in \mathcal{C}$ be any component.
We distinguish \emph{big} and \emph{small} terminals in $c$ depending on their demands.

\begin{definition}[big and small terminals]
\label{def:big-small}
Let $\alpha=\eps^{(1/\eps)+1}$.
Let $\Gamma'=\eps\cdot \alpha/\Gamma$, where $\Gamma$ is defined in \cref{lem:decomposition}.
We say that a terminal $v$ is \emph{big} if $\demand(v)>\Gamma'$ and \emph{small} otherwise.
\end{definition}

We partition the \emph{small} terminals in $c$ using a \emph{multi-level decomposition}: first, the component $c$ is decomposed into \emph{blocks} (\cref{sec:decomposition-level-1}); next, each block is decomposed into \emph{clusters} (\cref{sec:decomposition-level-2}); and finally, each cluster is decomposed into \emph{cells} (\cref{sec:decomposition-level-3}).

We introduce some common notations for blocks, clusters, and cells.
Each block (resp.\ cluster or cell) has a \emph{root} vertex and at most one \emph{exit} vertex.
We say that a terminal $v$ is \emph{strictly} inside a block (resp.\ cluster or cell) if $v$ belongs to the block (resp.\ cluster or cell) and, in addition, $v$ is different from the root vertex and the exit vertex of the block (resp.\ cluster or cell).
Note that any terminal strictly inside a block (resp.\ cluster or cell) is small.
The \emph{demand} of a block (resp.\ cluster or cell) is defined as the total demand of all terminals \emph{strictly} inside that block (resp.\ cluster or cell).
We say that a block (resp.\ cluster or cell) is \emph{passing} if it has an exit vertex and is \emph{ending} otherwise.
The \emph{spine} of a passing block (resp.\ passing cluster or passing cell) is the path between the root vertex and the exit vertex of that block (resp.\ cluster or cell).

\subsection{Decomposition of a Component Into Blocks (\cref{fig:component-block})}
\label{sec:decomposition-level-1}
    Let $c$ be any component.
    Let $U\subseteq V$ denote the set of vertices consisting of the big terminals in $c$, the \emph{root} vertex of $c$, and possibly the \emph{exit} vertex of $c$ if $c$ is an \emph{internal} component (see \cref{lem:decomposition} for definitions).
    Let $T_U$ denote the subtree of $c$ spanning the vertices in $U$.
    We say that a vertex in $T_U$ is a \emph{key} vertex if either it belongs to $U$ or it has two children in $T_U$.
    We define a \emph{block} to be a maximally connected subgraph of component $c$ in which any key vertex has degree 1; in other words, blocks are obtained by splitting the component at the key vertices.
    The blocks form a partition of the edges of component $c$.

\subsection{Decomposition of a Block Into Clusters (\cref{fig:block-cluster})}
\label{sec:decomposition-level-2}
    As an adaptation from \cref{lem:decomposition}, we decompose a block into clusters in \cref{lem:cluster-decomposition}.

\begin{lemma}
\label{lem:cluster-decomposition}
Let $b$ be any block.
There is a polynomial time algorithm to compute a partition of the edges of the block $b$ into a set of \emph{clusters}, such that all of the following properties are satisfied:
\vspace{-0.5em}
\begin{enumerate}
\itemsep-0.2em
    \item Every cluster $x$ is a connected subgraph of $b$;
    the \emph{root} vertex of the cluster $x$, denoted by $r_x$, is the vertex in $x$ that is closest to the depot.
    \item A cluster $x$ shares vertices with other clusters at vertex $r_x$ and possibly at one other vertex, called the \emph{exit} vertex of the cluster $x$ and denoted by $e_x$.
    If block $b$ has an exit vertex $e_b$, then there is a cluster $x$ in $b$ such that $e_x=e_b$.
    \item The demand of each cluster in $b$ is at most $2\Gamma'$.
    \item The number of clusters in $b$ is at most $3\cdot (\demand(b)/\Gamma'+1)$.
\end{enumerate}
\end{lemma}
\noindent The proof of \cref{lem:cluster-decomposition} is in \cref{sec:cluster-decomposition}.

\subsection{Decomposition of a Cluster Into Cells (\cref{fig:cluster-cell})}
\label{sec:decomposition-level-3}
    Let $x$ be any cluster.
    If $x$ is an ending cluster, then the decomposition of $x$ consists of a single \emph{cell}, which is the entire cluster $x$.
    If $x$ is a passing cluster, then we decompose $x$ into \emph{cells} as follows.
    Let $\ell$ denote the cost of the spine of  cluster $x$.
    When $\ell=0$, the decomposition of $x$ consists of a single \emph{cell}, which is the entire cluster $x$.
    Next, we assume that $\ell>0$.
    For each integer $i\in[1,(1/\eps)-1]$, there exists a unique edge $(u,v)$ on the spine of cluster $x$ satisfying $\min(\dist(r_x,u),\dist(r_x,v))\leq i\cdot \eps \cdot\ell< \max(\dist(r_x,u),\dist(r_x,v))$; let $e_i$ denote that edge.
    Removing the edges $e_1,e_2,\dots, e_{(1/\eps)-1}$ from cluster $x$ results in at most $1/\eps$ connected subgraphs; each subgraph is called a \emph{cell}.
    Observe that those cells form a partition of the vertices of cluster $x$.

    From the construction, the (unique) cell inside an ending cluster is an ending cell, and the cells inside a passing cluster are passing cells.
    The cost of the spine of any passing cell in a passing cluster is at most an $\eps$ fraction of the cost of the spine of that cluster.

\begin{fact}
\label{fact:constant-number}
In any component $c$, the number of cells and the number of big terminals are both $O_\eps(1)$.
\end{fact}

\begin{proof}
By \cref{lem:decomposition}, the total demand of the terminals in component $c$ is at most $2\Gamma$.
Since the demand of a big terminal is at least $\Gamma'$, there are at most $2\Gamma/\Gamma'=O_\eps(1)$ big terminals in $c$.

From the construction in \cref{sec:decomposition-level-1}, the set $U$ consists of at most $2+2\Gamma/\Gamma'$ vertices.
Since each vertex in $c$ has at most two children, the number of blocks in $c$ is at most $2|U|\leq 4+4\Gamma/\Gamma'$.
From the construction in \cref{sec:decomposition-level-2}, each block $b$ is partitioned into at most $3 \cdot(\demand(b)/\Gamma'+1)$ clusters, where $\demand(b)$ is at most the total demand of the terminals in component $c$, which is at most $2\Gamma$.
From the construction in \cref{sec:decomposition-level-3}, each cluster is partitioned into at most $1/\eps$ cells. So the number of cells in $c$ is at most
$\left(4+4\Gamma/\Gamma'\right)\cdot (3 \cdot (2\Gamma/\Gamma'+1))\cdot (1/\eps)=O_\eps (1).$
\end{proof}

\begin{definition}[Adaptation from \cref{def:subtour-component}]
\label{def:subtour-clusters-cells}
A \emph{subtour in a cluster (resp.\ cell)} is a path $t$ that starts and ends at the root of that cluster (resp.\ cell), and such that every vertex on $t$ is in that cluster (resp.\ cell).
We say that $t$ is a \emph{passing subtour} if that cluster (resp.\ cell) has an exit vertex and that vertex belongs to $t$, and is an \emph{ending subtour} otherwise.
The \emph{spine subtour} in a passing cluster (resp.\ passing cell) consists of the spine of that cluster (resp.\ cell) in both directions.
\end{definition}


\section{Simplifying the Local Solution}
\label{sec:local}

In this section, we prove the Local Theorem (\cref{thm:local}).

\begin{theorem}[Local Theorem]
\label{thm:local}
Let $c$ be any component.
Let $S_c$ denote a set of at most $(2\Gamma/\alpha)+1$ subtours in component $c$ visiting all terminals in $c$.
Then there exists a set $S^*_c$ of subtours in component $c$ visiting all terminals in $c$, such that all of the following properties hold:
\begin{enumerate}
\item For each cell in $c$, a single subtour in $S^*_c$ visits all small terminals in that cell;
\item $S^*_c$ contains one particular subtour $\bar{t}$ of demand at most 1, and the subtours in $S^*_c\setminus \{ \bar{t}\}$ are in one-to-one correspondence with the subtours in $S_c$,
    such that for every subtour $t$ in $S_c$ and its corresponding subtour $t^*$ in $S^*_c\setminus \{ \bar{t}\}$, the demand of $t^*$ is at most the demand of $t$, and in addition, if $t$ is a passing subtour in $c$, then $t^*$ is also a passing subtour in $c$;
\item The cost of $S^*_c$ is at most $1.5 + 2\eps$ times the cost of $S_c$.
\end{enumerate}
\end{theorem}

\subsection{Construction of $S^*_c$}
\label{sec:S*c}

The construction of $S^*_c$ starts from $S_c$ and proceeds in 5 steps.
Step 2 uses a new concept of \emph{threshold cells} and is the main novelty in the construction.
Step 1 and Step 3 are based on the following lemma due to Becker and Paul~\cite{becker2019framework}.

\begin{lemma}[Assignment Lemma, Lemma~1 in \cite{becker2019framework}]
\label{lem:assignment}
Let $G=(V[G],E[G])$ be an edge-weighted bipartite graph with vertex set $V[G]=A\uplus B$ and edge set $E[G]\subseteq A\times B$, such that each edge $(a,b)\in E[G]$ has a weight $w(a,b)\geq 0$.
For each vertex $b\in B$, let $N(b)$ denote the set of vertices $a\in A$ such that $(a,b)\in E[G]$.
We assume that $N(b)\neq \emptyset$ and the weight $w(b)$ of the vertex $b$ satisfies $0\leq w(b)\leq \sum_{a\in N(b)}w(a,b)$.
Then there exists a function $f: B\to A$ such that each vertex $b\in B$ is assigned to a vertex $a\in N(b)$ and, for each vertex $a\in A$, we have
\[\sum_{b\in B\mid f(b)=a} w(b) - \sum_{b\in B\mid (a,b)\in E[G]} w(a,b)\leq \max_{b\in B} \big\{w(b)\big\}.\]
\end{lemma}


\paragraph{Step 1: Combining ending subtours within each cluster.}
Let $A_0$ denote $S_c$.
We define a weighted bipartite graph $G$ in which the vertices in one part represent the subtours in $A_0$ and the vertices in the other part represent the clusters in $c$.\footnote{With a slight abuse, we identify a vertex in $G$ with either a subtour in $A_0$ or a cluster in $c$.}
There is an edge in $G$ between a subtour $a\in A_0$ and a cluster $x$ in $c$ if and only if $a$ contains an ending subtour $t$ in $x$; the weight of the edge is defined to be $\demand(t)$.
For each cluster $x$ in $c$, we define the weight of $x$ in $G$ to be the sum of the weights of its incident edges in $G$.
We apply the Assignment Lemma (\cref{lem:assignment}) to the graph $G$ (deprived of the vertices of degree 0) and obtain a function $f$ that maps each cluster $x$ in $c$ to some subtour $a\in A_0$ such that $(a,x)$ is an edge in $G$.

We construct a set of subtours $A_1$ as follows: for every cluster $x$ in $c$ and for every subtour $a\in A_0$ containing an ending subtour $t$ in $x$, the subtour $t$ is removed from $a$ and added to the subtour $f(x)$.
Observe that each resulting subtour in $A_1$ is connected.
From the construction, \emph{for each cluster $x$, at most one subtour in $A_1$ has an ending subtour in $x$}.
In particular, for any \emph{ending cell}, which is  equivalent to an ending cluster, a single subtour in $A_1$ visits all small terminals in that cell.

\paragraph{Step 2: Extending ending subtours within threshold cells.}

Let $x$ be any passing cluster in $c$ such that there is a subtour in $A_1$ containing an ending subtour in $x$.
From Step~1 of the construction, such a subtour in $A_1$ is unique; let $t_e$ denote the corresponding ending subtour in $x$.
Define the \emph{threshold cell} of cluster $x$ to be the deepest cell in $x$ containing vertices of $t_e$. See \cref{fig:threshold_before}.
We add to $t_e$ the part of the \emph{spine subtour in the threshold cell of $x$} that does not belong to $t_e$, resulting in a subtour $\tilde t_e$; see \cref{fig:threshold_after}.

Let $A_2$ denote the resulting set of subtours in $c$ after the extension within all threshold cells.
From the construction, \emph{for each passing cell $s$, all subtours in $s$ that are contained in $A_2$ are passing subtours in $s$.}

\begin{figure}[t]
\centering
    \begin{subfigure}{0.45\textwidth}
    \centering
    \includegraphics[scale=0.35]{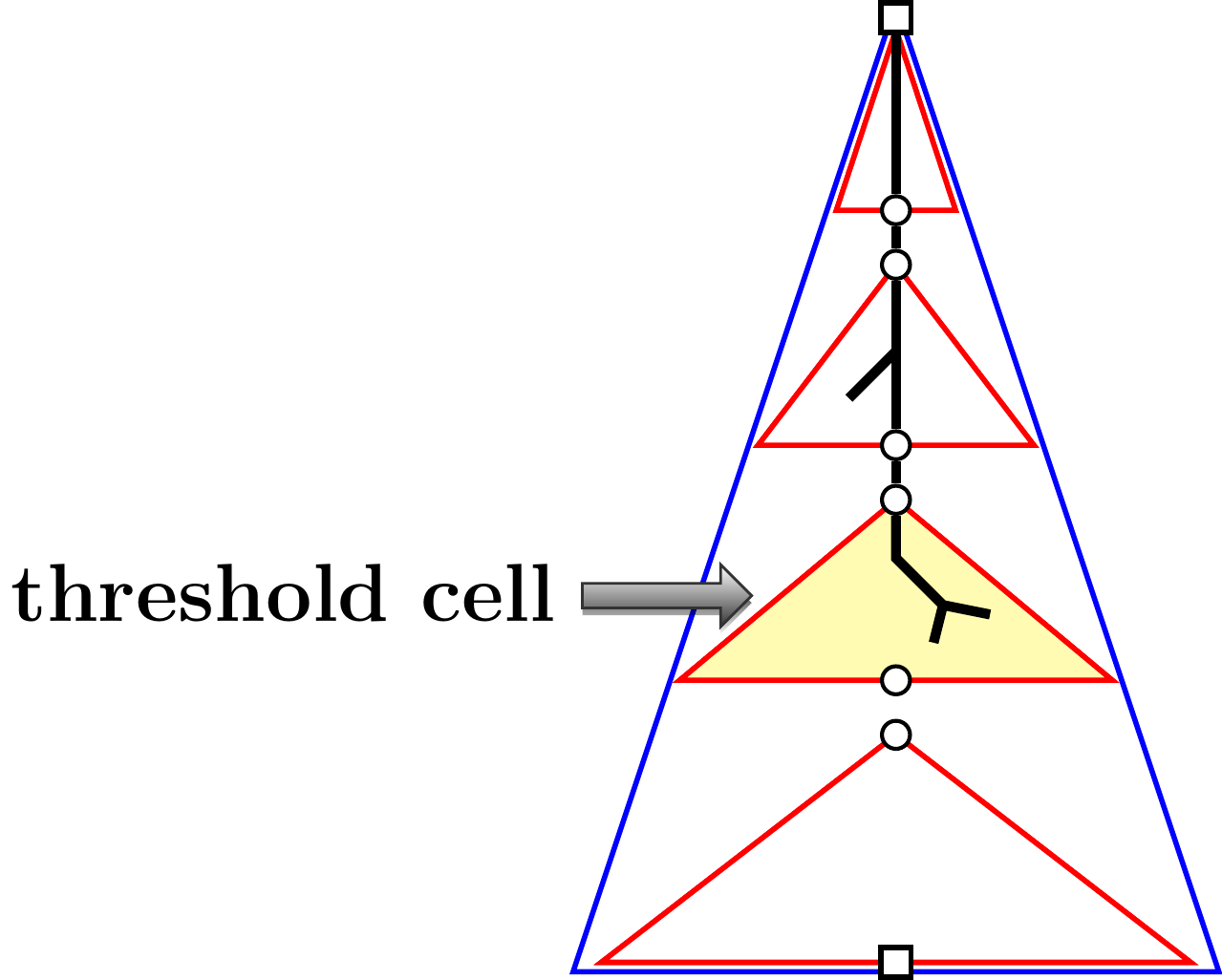}
    \caption{Subtour before extension.}
    \label{fig:threshold_before}
    \end{subfigure}
    \hfill
    \begin{subfigure}{0.45\textwidth}
    \centering
    \includegraphics[scale=0.35]{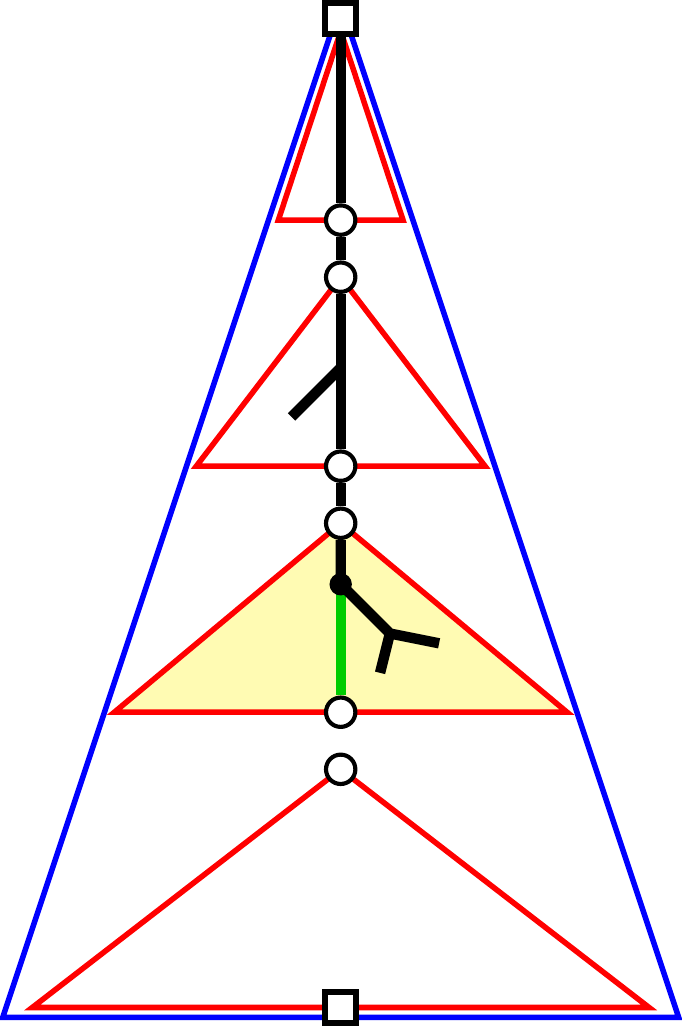}
    \caption{Subtour after extension.}
    \label{fig:threshold_after}
    \end{subfigure}
    \caption{\small The threshold cell and the extension of an ending subtour.
    The outermost triangle in blue represents a cluster $x$.
    In \cref{fig:threshold_before}, the black segments represent the ending subtour $t_e$ in $x$.
    The \emph{threshold cell} of cluster $x$ is the \emph{deepest} cell visited by $t_e$ and is represented by the yellow triangle.
    In \cref{fig:threshold_after}, subtour $t_e$ is extended within the threshold cell: the green segment represents the part of the \emph{spine subtour of the threshold cell} that is added to $t_e$, resulting in a subtour $\tilde t_e$.}
    \label{fig:threshold}
\end{figure}

\paragraph{Step 3: Combining passing subtours within each passing cell.}
We define a weighted bipartite graph $G'$ in which the vertices in one part represent the subtours in $A_2$ and the vertices in the other part represent the passing  cells in $c$.\footnote{With a slight abuse, we identify a vertex in $G'$ with either a subtour in $A_2$ or a passing cell in $c$.}
There is an edge in $G'$ between a subtour $a\in A_2$ and a passing cell $s$ in $c$ if and only if $a$ contains a non-spine passing subtour $t$ in $s$;  the weight of the edge is defined to be the total demand of the small terminals on $t$.
For each passing cell $s$ in $c$, we define the weight of $s$ in $G'$ to be the sum of the weights of its incident edges in $G'$.
We apply the Assignment Lemma (\cref{lem:assignment}) to the graph $G'$ (deprived of the vertices of degree 0) and obtain a function $f'$ that maps each passing cell $s$ in $c$ to some subtour $a\in A_2$ such that $(a,s)$ is an edge in $G'$.

We construct a set of subtours $A_3$ as follows:
for every passing cell $s$ in $c$ and for every subtour $a\in A_2$ containing a non-spine passing subtour $t$ in $s$, the subtour $t$ is removed from $a$ except for the spine subtour of $s$; the removed part is added to the subtour $f'(s)$.
Observe that each resulting subtour in $A_3$ is connected.
From the construction, \emph{for each passing cell $s$, a single subtour in $A_3$ visits all small terminals in $s$}.

\paragraph{Step 4: Correcting subtour capacities.}
For each subtour $t_3$ in $A_3$, let $t_0$ denote the corresponding subtour in $A_0$.
As soon as the demand of $t_3$ is greater than the demand of $t_0$, we repeatedly modify $t_3$ as follows: find a terminal $v$ that is \emph{visited by $t_3$ but not visited by $t_0$}; let $s$ denote the cell containing $v$ and let $t_s$ denote the subtour of $t_3$ in cell $s$;
if $s$ is an ending cell, then remove $t_s$ from $t_3$; and if $s$ is a passing cell, then remove $t_s$ from $t_3$ except for the spine subtour of~$s$.

Let $A_4$ denote the resulting set of modified subtours.
Observe that each subtour in $A_4$ is connected.
From the construction, \emph{the demand of each subtour in $A_4$ is at most the demand of the corresponding subtour in $A_0$.}
Note that the big terminals in each subtour in $A_4$ are the same as the big terminals in the corresponding subtour in $A_0$.\footnote{Any big terminal cannot be removed, since it is the exit vertex of some cell, thus belongs to the spine of that cell.
}

Let $\mathcal{R}$ denote the set of the removed pieces.
We claim that the total demand of the pieces in $\mathcal{R}$ is at most 1 (\cref{lem:feasible}).

\paragraph{Step 5: Creating an additional subtour.}
We connect all pieces in $\mathcal{R}$ by a single subtour in component $c$; let $\bar{t}$ be that subtour.

\vspace{2mm}
Finally, let $S^*_c$ denote $A_4\cup \{\bar{t}\,\}$.

\subsection{Analysis on the Cost of $S^*_c$}
\label{sec:analysis-S*c}
From the construction of $S_c^*$, we observe that the cost of $S_c^*$ equals the cost of $S_c$ plus the extra costs in Step~2 and in Step~5 of the construction, denoted by $W_2$ and $W_5$, respectively.

To analyze the extra costs, first, in a preliminary lemma (\cref{lem:threshold-cells}), we bound the overall cost of the spines of the threshold cells.
\cref{lem:threshold-cells} will be used to analyze both $W_2$ (\cref{cor:extra-cost-1}) and $W_5$ (\cref{lem:extra-cost-2}).
\begin{lemma}
\label{lem:threshold-cells}
The overall cost of the spines of all threshold cells in the component $c$ is at most $(\eps/2)\cdot \cost(S_c)$.
\end{lemma}

\begin{proof}
Consider any threshold cell $s$.
Let $x$ be the passing cluster that contains $s$.
As observed in \cref{sec:decomposition-level-3}, the cost of the spine of cell $s$ is at most an $\eps$ fraction of the cost of the spine of the cluster $x$.
Since $x$ is a passing cluster, at least one subtour in $S_c$ contains a passing subtour in $x$; let $t_x$ denote that passing subtour in $x$.
Observe that $t_x$ contains each edge of the spine of cluster $x$ in both directions (\cref{def:subtour-clusters-cells}), so the cost of the spine of $x$ is at most $\cost(t_x)/2$.
Thus the cost of the spine of $s$ is at most $(\eps/2)\cdot\cost(t_x)$.
We \emph{charge} the cost of the spine of $s$ to $t_x$.

From the construction, each cluster contains at most one threshold cell.
Thus the costs of the spines of all threshold cells are charged to disjoint parts of $S_c$.
The claim follows.
\end{proof}

Observe that the extra cost in Step~2 of the construction is at most the overall cost of the spine subtours in all threshold cells in the component $c$, which equals twice the overall cost of the spines of those cells by \cref{def:subtour-clusters-cells}.

\begin{corollary}
\label{cor:extra-cost-1}
The  extra cost $W_2$ in Step~2 of the construction is at most $\eps\cdot\cost(S_c)$.
\end{corollary}

Next, we bound the extra cost in Step~5 of the construction.


\begin{fact}
\label{fact:same-type}
Let $t$ denote any subtour in $S_c$.
Let $x$ denote any cluster in $c$.
Let $r_c$ and $r_x$ denote the root vertices of component $c$ and of cluster $x$, respectively;
let $e_x$ denote the exit vertex of cluster $x$.
If the $r_c$-to-$r_x$ path (resp.\ the $r_c$-to-$e_x$ path) belongs to $t$, then that path belongs to the corresponding subtour of $t$ throughout the construction in \cref{sec:S*c}.
\end{fact}

\begin{definition}[nice edges]
We say that an edge $e$ in component $c$ is \emph{nice} if $e$ belongs to \emph{at least two} subtours in $A_2$.
\end{definition}

The next Lemma is the main novelty in the analysis.

\begin{lemma}
\label{lem:connected}
Any piece in $\mathcal{R}$ is connected to the root $r_c$ of component $c$ through nice edges in $c$.
\end{lemma}

\begin{proof}
Consider any piece  $q\in \mathcal{R}$.
Let $s$ be the cell containing $q$.
Let $x$ be the cluster containing $q$.
See \cref{fig:analysis}.
Let $r_s$ and $r_x$ denote the root vertices of cell $s$ and of cluster $x$, respectively.
Observe that the terminals in $x$ are visited by at least two subtours in $S_c$.
This is because, if all terminals in cluster $x$ are visited by a single subtour in $S_c$, then those terminals belong to the corresponding subtour throughout the construction, thus none of those terminals belongs to a piece in $\mathcal{R}$, contradiction.
Thus the $r_c$-to-$r_x$ path belongs to at least two subtours in $S_c$.
By \cref{fact:same-type}, the $r_c$-to-$r_x$ path belongs to at least two subtours in $A_2$, thus every edge on the $r_c$-to-$r_x$ path is nice.
It suffices to show the following Claim:
\begin{center}
\emph{Piece $q$ is connected to vertex $r_x$ through nice edges in $c$.\qquad (*)}
\end{center}

\noindent There are two cases:

\emph{Case 1: $x$ is an ending cluster.} See \cref{fig:analysis-1}.
From the decomposition in \cref{sec:decomposition-level-3}, $s$ is an ending cell and $s$ equals $x$.
Piece $q$ is an ending subtour in $x$ and in particular contains $r_x$.
Claim (*) follows trivially.

\emph{Case 2: $x$ is a passing cluster.}
Let $e_s$ and $e_x$ denote the exit vertices of cell $s$ and of cluster $x$, respectively.
Observe that at least one subtour in $S_c$ contains a passing subtour in $x$.
There are two subcases.

\emph{Subcase 2(i): At least two subtours in $S_c$ contain passing subtours in $x$.} See \cref{fig:analysis-2}.
Then the $r_c$-to-$e_x$ path belongs to at least two subtours in $S_c$.
By \cref{fact:same-type}, the $r_c$-to-$e_x$ path belongs to at least two subtours in $A_2$, thus each edge on the spine of $x$ is nice.
Since piece $q$ contains a vertex on the spine of $x$, Claim (*) follows.

\emph{Subcase 2(ii): Exactly one subtour in $S_c$ contains a passing subtour in $x$.} See \cref{fig:analysis-3,fig:analysis-4}.
Let $t_p$ denote that passing subtour in $x$.
As observed previously, at least two subtours in $S_c$ visit terminals in $x$, so there must be at least one subtour in $S_c$ that contains an ending subtour in $x$.
Let $t_e^{1},\dots, t_e^{m}$ (for some $m\geq 1$) denote the ending subtours in $x$ contained in the subtours in $S_c$.
In Step~1 of the construction, the $m$ ending subtours are combined into a single ending subtour, denoted by $t_e$ (recall that the threshold cell of $x$ is defined with respect to $t_e$);
and in Step~2 of the construction, subtour $t_e$ is extended to a subtour $\tilde t_e$ (\cref{fig:threshold}).
Note that the passing subtour $t_p$ remains unchanged in Steps~1~and~2 of the construction.
We observe that cell $s$ is either above or equal to the threshold cell of $x$.
This is because, if cell $s$ is below the threshold cell of $x$, then all terminals in $s$ are visited by a single subtour in $S_c$, i.e.,\ the subtour $t_p$, so those terminals belong to the corresponding subtour of $t_p$ throughout the construction, thus none of those terminals belongs to a piece in $\mathcal{R}$, contradiction.
Hence the following two subsubcases.

\emph{Subsubcase~2(ii)($\alpha$): $s$ is above the threshold cell of $x$}, see \cref{fig:analysis-3}.
Each edge on the $r_x$-to-$e_s$ path belongs to both subtours $t_p$ and $t_e$, hence is nice.
Since $q$ contains some vertex on the spine of $s$, Claim (*) follows.

\emph{Subsubcase~2(ii)($\beta$): $s$ equals the threshold cell of $x$}, see \cref{fig:analysis-4}.
Observe that each edge on the $r_x$-to-$e_s$ path belongs to $\tilde t_e$ due to the extension of the ending subtour $t_e$ within the threshold cell (Step~2 of the construction).
Thus each edge on the $r_x$-to-$e_s$ path belongs to both subtours $t_p$ and $\tilde t_e$, hence is nice.
Since $q$ contains some vertex on the spine of $s$, Claim (*) follows.
\end{proof}

\begin{figure}[t]
\centering
\vspace{20mm}
    \begin{subfigure}[t]{0.24\textwidth}
    \centering
    \includegraphics[scale=0.35]{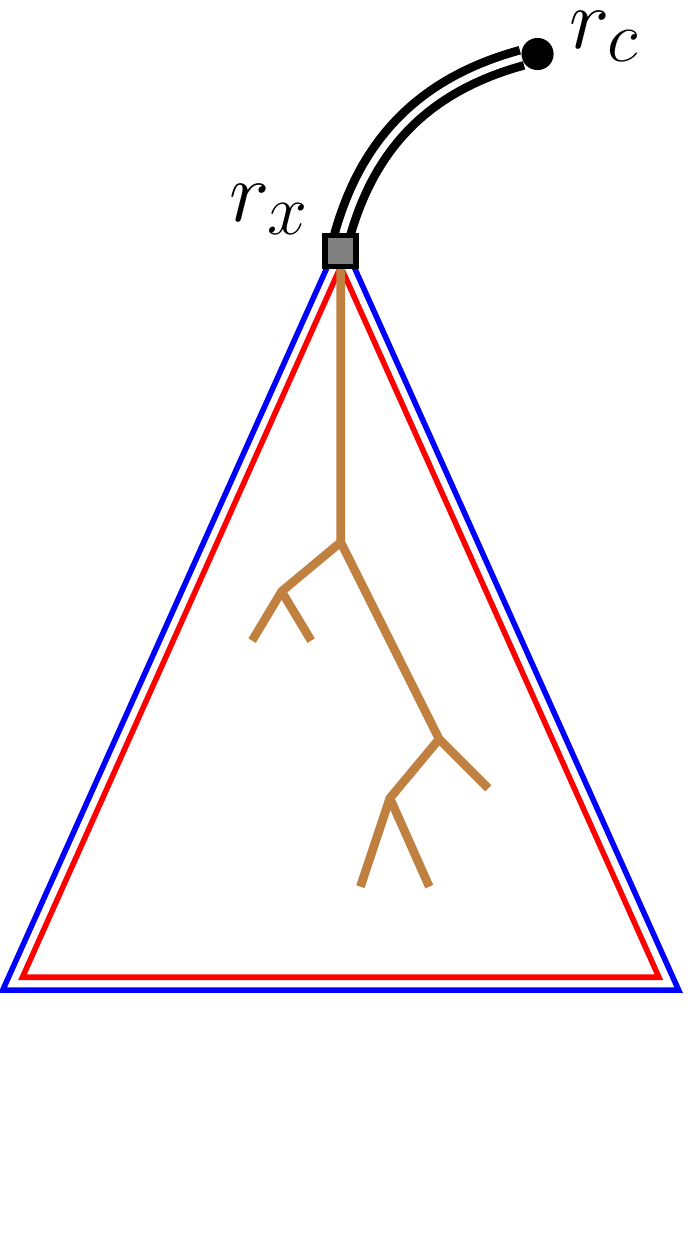}
    \caption{Case 1.}
    \label{fig:analysis-1}
    \end{subfigure}
    \begin{subfigure}[t]{0.24\textwidth}
    \centering
    \includegraphics[scale=0.35]{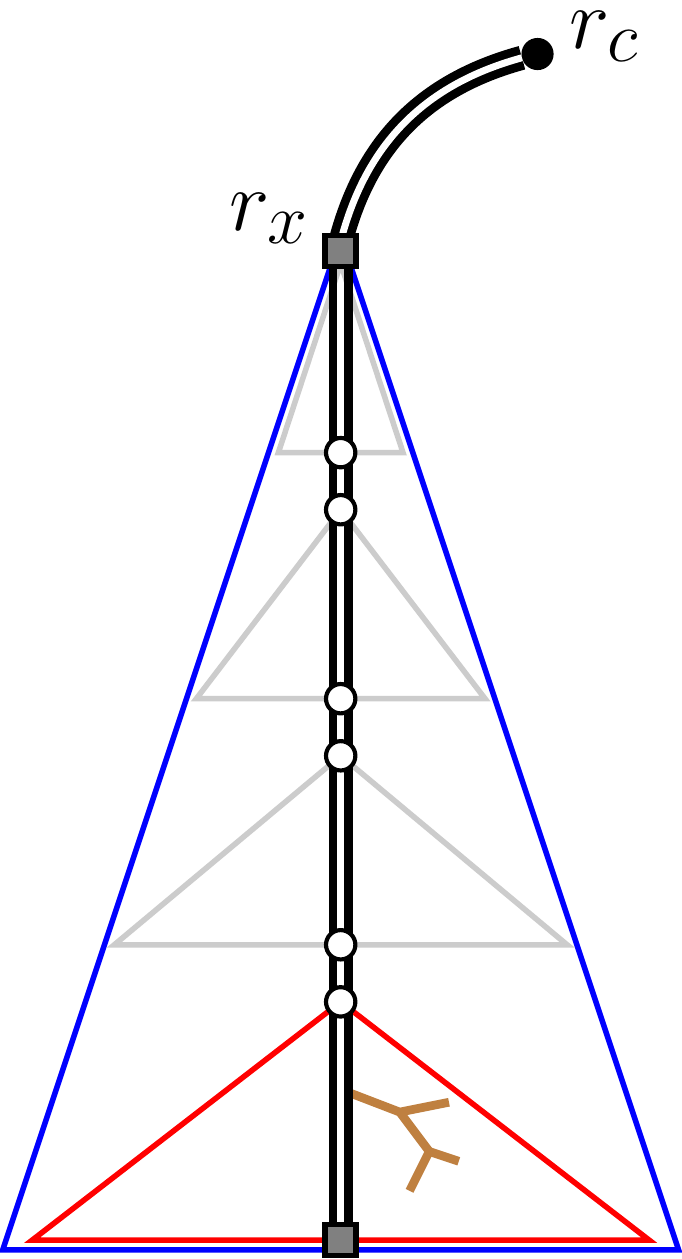}
    \caption{Subcase 2(i).}
    \label{fig:analysis-2}
    \end{subfigure}
    \hfill
    \begin{subfigure}[t]{0.24\textwidth}
    \centering
    \includegraphics[scale=0.35]{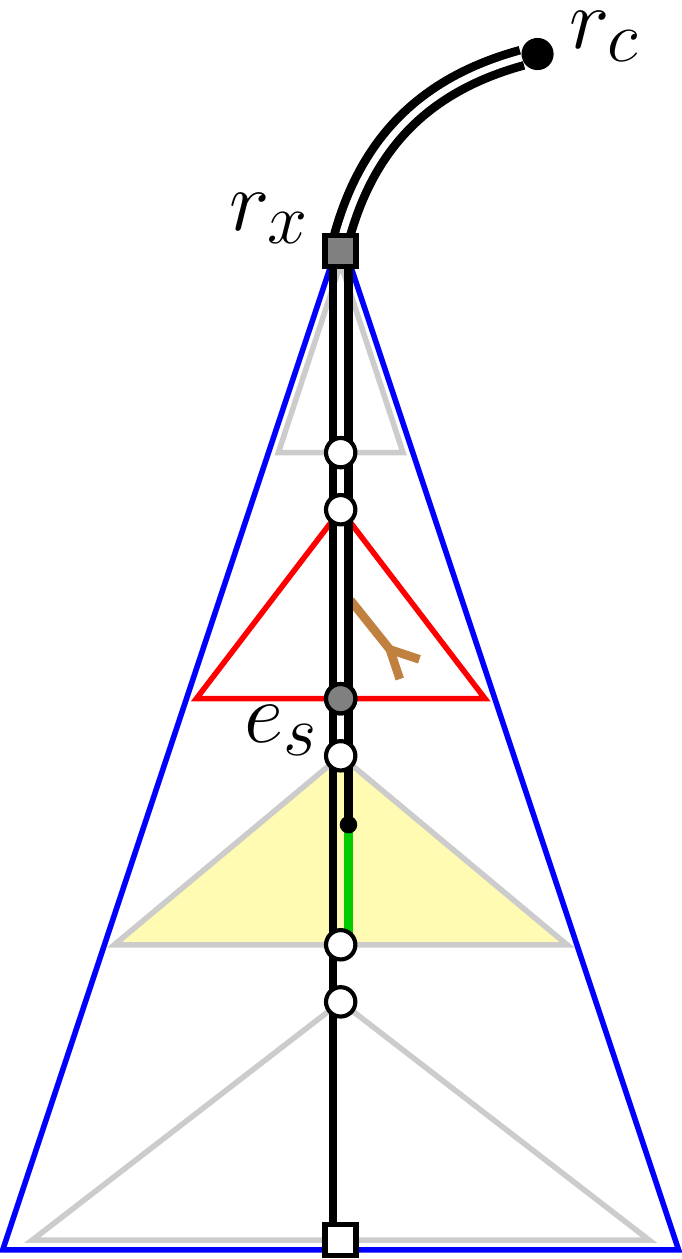}
    \caption{Subsubcase 2(ii)($\alpha$).}
    \label{fig:analysis-3}
    \end{subfigure}
    \hfill
    \begin{subfigure}[t]{0.24\textwidth}
    \centering
    \includegraphics[scale=0.35]{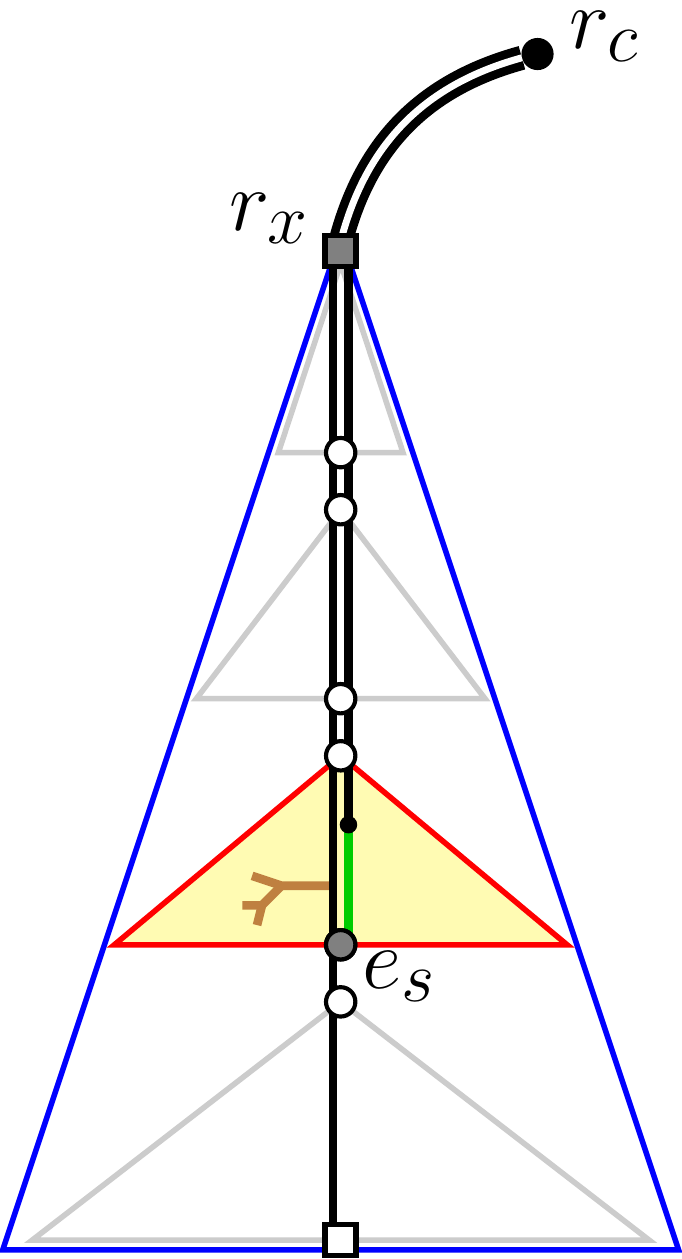}
    \caption{Subsubcase 2(ii)($\beta$)}
    \label{fig:analysis-4}
    \end{subfigure}
    \caption{\small Illustrations for the different cases in the proof of \cref{lem:connected}.
    A removed piece $q\in \mathcal{R}$ is in brown.
    The cell $s$ containing that piece is represented by the triangle in red; the cluster $x$ containing that piece is represented by the outermost triangle in blue.
    The black node $r_c$ is the root of component $c$.
    In \cref{fig:analysis-1}, $x$ is an ending cluster.
    In \cref{fig:analysis-2}, $x$ is a passing cluster, and the solution $S_c$ contains two passing subtours in $x$.
    In \cref{fig:analysis-3,fig:analysis-4}, $x$ is a passing cluster, and the solution $S_c$ contains a unique passing subtour in $x$; the yellow triangle represents the threshold cell of $x$.
    In the case when $q$ belongs to the threshold cell (\cref{fig:analysis-4}), $q$ is connected to $r_c$ through at least two subtours, thanks to the extension of the ending subtour within the threshold cell.
    }
    \label{fig:analysis}
\end{figure}

\begin{lemma}
\label{lem:extra-cost-2}
The extra cost $W_5$ in Step~5 of the construction is at most $(0.5+\eps)\cdot \cost(S_c)$.
\end{lemma}

\begin{proof}
First, we argue that the extra cost $W_5$ in Step~5 of the construction is at most twice the overall cost of the nice edges in $c$.
Let $H$ be the multi-subgraph in $c$ that consists of the pieces in $\mathcal{R}$ and two copies of each nice edge in $c$ (one copy for each direction).
Since any piece in $\mathcal{R}$ is connected to the root $r_c$ of component $c$ through nice edges (\cref{lem:connected}), $H$ induces a connected subtour in $c$.

Second, we bound the overall cost of the nice edges in $c$.
From the construction, any nice edge $e$ in $c$ is of at least one of the two classes:
\begin{enumerate}
\item edge $e$ belongs to at least two subtours in $S_c$;
\item edge $e$ belongs to the spine of a threshold cell in component $c$.
\end{enumerate}
Each nice edge $e$ of the first class has at least 4 copies in $S_c$, since each subtour to which $e$ belongs contains 2 copies of $e$ (one for each direction).
Thus the overall cost of the nice edges of the first class is at most $0.25\cdot \cost(S_c)$.
By \cref{lem:threshold-cells}, the overall cost of the nice edges of the second class is at most $(\eps/2)\cdot\cost(S_c)$.
Hence the overall cost of the nice edges of both classes is at most $(0.25+\eps/2)\cdot \cost(S_c)$.

Since the extra cost $W_5$ is at most twice the overall cost of the nice edges in $c$, we have $W_5\leq (0.5+\eps)\cdot \cost(S_c)$.
\end{proof}

From \cref{cor:extra-cost-1,lem:extra-cost-2}, we conclude that
\[\cost(S_c^*)= \cost(S_c)+W_2+W_5\leq (1.5+2\eps)\cdot \cost(S_c).\]
Hence the third property of the claim in the Local Theorem (\cref{thm:local}).

\subsection{Feasibility}

From the construction, $S_c^*$ is a set of subtours in $c$ visiting all terminals in $c$.
The first property of the claim in the Local Theorem (\cref{thm:local}) follows from the construction.
The second property of the claim follows from the construction, \cref{fact:same-type}, and the following \cref{lem:feasible}.

\begin{lemma}
\label{lem:feasible}
The total demand of the pieces in $\mathcal{R}$ is at most 1.
\end{lemma}

\begin{proof}
Observe that the pieces in $\mathcal{R}$ are removed from subtours in $A_3$.
Let $t_3$ denote any subtour in $A_3$.
Let $t_0$, $t_1$, $t_2$, and $t_4$ denote the corresponding subtours of $t_3$ in $A_0$, $A_1$, $A_2$, and $A_4$, respectively.
Let $\Delta$ denote the overall demand of the pieces that are removed from $t_3$ in Step~4 of the construction.
Observe that $\Delta=\demand(t_3)-\demand(t_4)$.
To bound $\Delta$, first, by Step~1 of the construction and the Assignment Lemma (\cref{lem:assignment}), $\demand(t_1)-\demand(t_0)$ is at most the maximum demand of a cluster, which is at most $2\Gamma'$ by the definition of clusters (\cref{sec:decomposition-level-2}).
By Step~2 of the construction, $\demand(t_2)=\demand(t_1)$.
By Step~3 of the construction and the Assignment Lemma (\cref{lem:assignment}), $\demand(t_3)-\demand(t_2)$ is at most the maximum demand of a cell, which is at most $2\Gamma'$ by the definition of cells (\cref{sec:decomposition-level-3}).
By Step~4 of the construction, $\demand(t_0)-\demand(t_4)$ is at most the maximum demand of a cell, which is at most $2\Gamma'$.
Combining, we have $\Delta=\demand(t_3)-\demand(t_4)\leq 6\Gamma'$.

The number of subtours in $A_3$ equals the number of subtours in $S_c$, which is at most $(2\Gamma/\alpha)+1$ by assumption.
Thus total demand of the pieces in $\mathcal{R}$ is at most $6\Gamma'\cdot ((2\Gamma/\alpha)+1)<13\eps<1$, assuming $\eps<1/13$.
\end{proof}

This completes the proof of the Local Theorem (\cref{thm:local}).

\section{Height Reduction}

\label{sec:hat-T}

In this section, we transform the tree $T$ into a tree $\hat T$ so that $\hat T$ has $O_\eps (1)$ levels of components.

All results in this section are already given in \cite{MZ22} for the equal demand setting.
The arguments for the arbitrary demand setting are identical, except for the proof of \cref{thm:MZ22-opt}, which is a minor adaptation of the proof in~\cite{MZ22}, see \cref{rmk:MZ22-opt}.

\begin{lemma}[Lemma 21 in \cite{MZ22}]
\label{lem:tild-D-H-eps}
Let $\tilde D=\alpha\cdot\eps\cdot D_{\min}$, where $\alpha$ is defined in \cref{def:big-small} and $D_{\min}$ is defined in \cref{def:bounded-distance}.
Let $H_\eps=(1/\eps)^{(2/\eps)+1}$.
For each $i\in[1,H_\eps]$, let $\mathcal{C}_i\subseteq \mathcal{C}$ denote the set of components $c\in \mathcal{C}$ such that $\dist(r,r_c)\in\left[(i-1)\cdot \tilde D, i\cdot \tilde D\right)$.
Then any component $c\in \mathcal{C}$ belongs to a set $\mathcal{C}_i$ for some $i\in [1,H_\eps]$.
\end{lemma}

\begin{definition}[Definition 22 in \cite{MZ22}]
\label{def:critical}
We say that a set of components $\tilde{\mathcal{C}}\subseteq \mathcal{C}_i$ is \emph{maximally connected} if the components in $\tilde{\mathcal{C}}$ are connected to each other and $\tilde{\mathcal{C}}$ is maximal within $\mathcal{C}_i$.
For a maximally connected set  of components $\tilde{\mathcal{C}}\subseteq \mathcal{C}_i$, we define the \emph{critical vertex} of $\tilde{\mathcal{C}}$ to be the root vertex of the component $c\in\tilde{\mathcal{C}}$ that is closest to the depot.
\end{definition}

\begin{algorithm}[H]
\caption{Construction of the tree $\hat T$ (\cite{MZ22}).}
\label{alg:hat-T}
\begin{algorithmic}[1]
\For{each $i\in[1,H_\eps]$}
    \For{each maximally connected set of components $\tilde{\mathcal{C}}\subseteq \mathcal{C}_i$}
        \State $z\gets$ critical vertex of $\tilde{\mathcal{C}}$
        \For{each component $c\in \tilde{\mathcal{C}}$}
            \State $\delta\gets r_c$-to-$z$ distance in $T$
            \State \emph{Split} the tree $T$ at the root vertex $r_c$ of the component $c$
            \State Add an edge between the root of the component $c$ and $z$ with weight $\delta$
        \EndFor
    \EndFor
\EndFor
\State $\hat T\gets$the resulting tree
\end{algorithmic}
\end{algorithm}

\begin{fact}[Fact 6 in \cite{MZ22}]
\label{fact:hat-T}
\cref{alg:hat-T} constructs in polynomial time a tree $\hat T$ such that:
\begin{itemize}
    \item The components in $\hat T$ are the same as those in the tree $T$;
    \item Any solution to the UCVRP on the tree $\hat T$ can be transformed in polynomial time into a solution to the UCVRP on the tree $T$ without increasing the cost.
\end{itemize}
\end{fact}

\begin{theorem}[Adaptation of Theorem 23 in \cite{MZ22}]
\label{thm:MZ22-opt}
Consider the UCVRP on the tree $\hat T$.
There exist dummy terminals of appropriate demands and a solution $\OPT_2$ visiting all of the real and the dummy terminals, such that all of the following properties hold:
\begin{enumerate}
\item For each component $c$, there are at most $(2\Gamma/\alpha)+1$ tours visiting terminals in $c$;
\item For each component $c$ and each tour $t$ visiting terminals in $c$, the total demand of the terminals in $c$ visited by $t$ is at least $\alpha$;
\item The cost of $\OPT_2$ is at most $1+3\eps$ times the optimal cost for the UCVRP on the tree $T$.
\end{enumerate}
\end{theorem}

\begin{remark}
\label{rmk:MZ22-opt}
In the proof of \cref{thm:MZ22-opt}, everything in~\cite{MZ22} carries over to the arbitrary demand setting, except that the \emph{Iterated Tour Partitioning (ITP)} algorithm, which is used to reduce the demand of the tours exceeding capacity, is adapted as follows.
Let $t_{\TSP}$ be a traveling salesman tour visiting a selected subset of terminals (Section~3.1 in \cite{MZ22}).
We partition $t_{\TSP}$ into segments, such that all segments, except possibly the last segment, have demands between $1-\alpha$ and $1$, instead of being exactly $1$.
This is achievable since every terminal on $t_{\TSP}$  has demand at most $\alpha$.
The analysis is identical to \cite{MZ22} except for a suitable adaptation of the analysis of the ITP algorithm.
\end{remark}

\section{Combining Local Solutions}
\label{sec:global}

\begin{definition}
\label{def:Y-c}
For any component $c\in\mathcal C$, let $\mathcal{Q}_c$ denote the partition of the terminals of component $c$, such that each part of the partition consists of either all small terminals in a cell in $c$, or a single big terminal in $c$.
We define the set \[Y_c=\{\alpha\}\cup\Big(\big\{ \demand\big(\tilde{\mathcal{Q}}_c\big): \tilde{\mathcal{Q}}_c\subseteq \mathcal{Q}_c\big\}\cap (\alpha,1]\Big).\]
\end{definition}

\begin{theorem}
\label{thm:global}
Consider the UCVRP on the tree $\hat T$.
There exist dummy terminals of appropriate demands and a solution $S^*$ visiting all of the real and the dummy terminals, such that all of the following properties hold:
\begin{enumerate}
    \item For each cell, a single tour visits all small terminals in that cell;

    \item For each component $c$ and each tour $t$ visiting terminals in $c$, the total demand of the terminals in $c$ visited by $t$ belongs to $Y_c$;
    \item The cost of $S^*$ is at most $1.5+7\eps$ times the optimal cost for the UCVRP on the tree $T$.
\end{enumerate}
\end{theorem}

\begin{proof}
Let $S$ denote the solution to the UCVRP on the tree $\hat T$ that is identical to $\OPT_2$ defined in \cref{thm:MZ22-opt}, except for ignoring dummy terminals in $\OPT_2$.
To construct $S^*$, we modify the solution $S$ component by component as follows.
Consider any component $c$.
Let $S_c$ denote the set of subtours in $c$ obtained by restricting the tours of $S$ visiting terminals in $c$.
By \cref{thm:MZ22-opt}, $|S_c|\leq (2\Gamma/\alpha)+1$.
We apply the Local Theorem (\cref{thm:local}) on $S_c$ to obtain a set $S_c^*$ of subtours in $c$, such that $S_c^*$ contains one particular subtour $\bar{t}$ of demand at most 1, and the subtours in $S_c^*\setminus\{\bar t\}$ are in one-to-one correspondence with the subtours in $S_c$.
For the subtour $\bar{t}$, we create a new tour from the depot by adding the $r$-to-$r_c$ connection in both directions.
Next, consider any subtour $t^{(1)}$ in $S_c$; let $t^{(2)}$ be the corresponding subtour in $S_c^*$.
We replace $t^{(1)}$ in the solution $S$ by a subtour $t$ defined as follows: $t$ is identical to $t^{(2)}$, except that if $t^{(2)}$ visits terminals in $c$ and the demand of $t^{(2)}$ is less than $\alpha$, then we add to $t$ a \emph{dummy} terminal at $r_c$  of demand $\alpha-\demand(t^{(2)})$.

Let $S^*$ denote the resulting solution.

By the second property of the Local Theorem (\cref{thm:local}), in any component $c$, if a subtour in $S_c$ is a passing subtour, then the corresponding subtour in $S_c^*$ is a passing subtour, thus each tour in $S^*$ is a connected tour starting and ending at the depot.
Again by the Local Theorem, in any component $c$, the subtours in $S_c^*$ together visit all terminals in $c$, thus the tours in $S^*$ together visit all of the real and the dummy terminals in $\hat T$.

To see that the tours in $S^*$ are within the capacity, first, by the second property of the Local Theorem (\cref{thm:local}), any additional tour created from a subtour $\hat t$ in some set $S_c^*$ is within the capacity.
Next, consider any tour $a$ among the remaining tours in $S^*$.
Consider any component $c$ that contains terminals visited by $a$; let $t$ be the subtour in $c$ from tour $a$.
Let $t^{(0)}$, $t^{(1)}$, and $t^{(2)}$ denote the corresponding subtours of $t$ in $\OPT_2$, $S_c$, and $S_c^*$, respectively.
By \cref{thm:MZ22-opt}, $\demand(t^{(0)})\geq \alpha$.
By definition of $S$, $\demand(t^{(1)})\leq \demand(t^{(0)})$.
Now we use the second property of the Local Theorem (\cref{thm:local}), which ensures that $\demand(t^{(2)})\leq \demand(t^{(1)})$.
Thus $\demand(t)= \max\{\demand(t^{(2)}),\alpha\}\leq \demand(t^{(0)})$.
Let $a^{(0)}$ denote the tour in $\OPT_2$ corresponding to $a$.
Then $\demand(a)\leq \demand(a_0)$.
Since $\OPT_2$ is a solution to the UCVRP (\cref{thm:MZ22-opt}), the total demand of (the real and the dummy) terminals in $a^{(0)}$ is at most 1.
Thus tour $a$ is within the capacity.

The first property of the claim follows from the first property of the Local Theorem (\cref{thm:local}).
The second property of the claim follows from the first property of the claim and the construction of $S^*$.
It remains to analyze the cost of $S^*$.

From the construction of $S^*$, we have
\begin{equation}
\label{eqn:S*}
\cost(S^*) = \sum_c \cost(S^*_c)+\sum_c 2\cdot \dist(r,r_c),
\end{equation}
where we use the fact that, for each component $c$, the distance in the tree $\hat T$ between the depot $r$ and the root $r_c$ of component $c$ equals that distance in the tree $T$.
By Lemma 17 in~\cite{MZ22},
\begin{equation}
    \label{eqn:MZ22-Lemma-17}
\sum_c\dist(r,r_c)\leq \frac{\eps}{8}\cdot \opt.
\end{equation}
By the third property of the Local Theorem (\cref{thm:local}), we have
\begin{equation}
\label{eqn:cost-Sc*}
\sum_c \cost(S^*_c)\leq \sum_c (1.5+2\eps)\cdot\cost(S_c)=(1.5+2\eps)\cdot\cost(S).
\end{equation}
By the definition of $S$ and \cref{thm:MZ22-opt}, we have
\begin{equation}
\label{eqn:cost-S}
\cost(S)= \cost(\OPT_2)\leq (1+3\eps)\cdot \opt,
\end{equation}
where $\opt$ denotes the optimal cost for the UCVRP on the tree $T$.

The last property of the claim follows from \cref{eqn:S*,eqn:cost-Sc*,eqn:cost-S,eqn:MZ22-Lemma-17}.
\end{proof}

\section{Structure Theorem}

\label{sec:structure}

\begin{definition}
\label{def:Y}
Let $Y\subseteq[\alpha,1]$ denote the set of values $y\in [\alpha,1]$ such that $y$ equals the sum of the elements in a multi-subset of $\bigcup_c Y_c$.
\end{definition}

\begin{fact}
\label{fact:Y}
$\{\mathcal{Q}_c\}_c$, $\{Y_c\}_c$, and $Y$ satisfy the following properties:
\begin{enumerate}
\item For any component $c$, the set $\mathcal{Q}_c$ consists of $O_\eps(1)$ parts and the set $Y_c$ consists of $O_\eps(1)$ values;
\item For any component $c$, we have $Y_c\subseteq Y$;
\item For any values $y\in Y$ and $y'\in Y$ such that $y+y'\leq 1$, we have $y+y'\in Y$;
\item The set $Y$ consists of $n^{O_\eps(1)}$ values.
\end{enumerate}
\end{fact}

\begin{proof}
The first property of the claim follows from \cref{fact:constant-number}.
The second and the third properties of the claim follow from \cref{def:Y}.
For any component $c$, each value in $Y_c$ is at least $\alpha$, so each value $y\in Y$ is the sum of at most $1/\alpha$ values in $\bigcup_c Y_c$.
Since the number of components in the tree $\hat T$ is at most $n$, the fourth property of the claim follows.
\end{proof}

\begin{definition}[subtours at a vertex]
For any vertex $v\in V$, we say that a path is \emph{a subtour at $v$} if that path starts and ends at $v$ and only visits vertices in the subtree of $\hat T$ rooted at $v$.
\end{definition}

We build on \cref{thm:global} to obtain the following Structure Theorem.

\begin{theorem}[Structure Theorem]
\label{thm:structure}
Let $\beta=\frac{1}{4}\cdot \eps^{(4/\eps)+1}$.
Consider the UCVRP on the tree $\hat T$.
There exist dummy terminals of appropriate demands and a solution $\hat S$ visiting all of the real and the dummy terminals, such that all of the following properties hold:
\begin{enumerate}
    \item For each cell, a single tour visits all small terminals in that cell;
    \item For each component $c$ and each tour $t$ visiting terminals in $c$, the total demand of the terminals in $c$ visited by $t$ belongs to $Y_c$;
    \item For the root of each component and for each critical vertex (\cref{def:critical}), the demand of each subtour at that vertex belongs to $Y$;
    \item For each critical vertex, there exist $\frac{1}{\beta}$ values in $Y$ such that demand of each  subtour at a child of that vertex is among those values;
    \item The cost of $\hat S$ is at most $1.5+8\eps$ times the optimal cost for the UCVRP on the tree~$T$.
\end{enumerate}
\end{theorem}

In the rest of the section, we prove the Structure Theorem (\cref{thm:structure}).

\subsection{Construction of $\hat S$}
\label{sec:construction-hat-S}
We construct the solution $\hat S$ by modifying the solution $S^*$ defined in \cref{thm:global}.
The construction is an adaptation of Section~5.1 in~\cite{MZ22}.

\begin{definition}
\label{def:I}
Let $I\subseteq V$ denote the set of vertices $v\in V$ such that $v$ is either the root of a component or a critical vertex.
\end{definition}

We consider the vertices in $I$ in the bottom up order.

Let $v$ be any vertex in $I$.
Let $S^*(v)$ denote the set of subtours at $v$ in $S^*$.
We construct a set $A(v)$ of subtours at $v$ satisfying the following invariants:
\begin{itemize}
\item the subtours in $A(v)$ have a one-to-one correspondence with the subtours in $S^*(v)$; and
\item for each subtour in $A(v)$, its demand belongs to $Y$ and is at most the demand of the corresponding subtour in $S^*(v)$.
\end{itemize}
The construction of $A(v)$ is according to one of the following three cases on $v$.

\paragraph{Case 1: $v$ is the root vertex $r_c$ of a leaf component $c$ in $\hat T$.}
Let $A(v)=S^*(v)$.
For each subtour in $A(v)$, its demand belongs to $Y_c$ by \cref{thm:global}, thus belongs to $Y$ by \cref{fact:Y}.


\paragraph{Case 2: $v$ is the root vertex $r_c$ of an internal component $c$ in $\hat T$.}
For each subtour $a\in S^*(v)$, if $a$ contains a subtour at the exit vertex $e_c$ of component $c$, letting $t$ denote this subtour and $t'$ denote the subtour in $A(e_c)$ corresponding to $t$, we replace the subtour $t$ in $a$ by the subtour $t'$. Let $a'$ denote the resulting subtour.
By induction, the demand of $t'$ is at most that of $t$, thus the demand of $a'$ is at most the demand of $a$.
Again by induction, the demand of $t'$ belongs to $Y$.
Using the same argument as before, the demand of the subtour in $c$ that is contained in $a$ belongs to $Y$.
Since the set $Y$ is closed under addition (\cref{fact:Y}), the demand of $a'$ is in $Y$.

Let $A(v)$ be the resulting set of subtours at $v$.


\paragraph{Case 3: $v$ is a critical vertex in $\hat T$.}

Let $r_1,\dots, r_m$ be the children of $v$ in $\hat T$.
For each subtour $a\in S^*(v)$ and for each $i\in[1,m]$, if $a$ contains a subtour at $r_i$, letting $t$ denote this subtour and $t'$ denote the subtour in $A(r_i)$ corresponding to $t$, we replace $t$ in $a$ by $t'$.
Let $A_1(v)$ denote the resulting set of subtours at $v$.

Let $W_v$ denote the set of the subtours at the children of $v$ in $A_1(v)$, i.e., $W_v=A(r_1)\cup\dots\cup A(r_m)$.
By induction, the demand of each subtour in $W_v$ belongs to $Y$.
If $|W_v|\leq \frac{1}{\beta}$, let $A(v)=A_1(v)$.

Next, consider the non-trivial case when $|W_v|>\frac{1}{\beta}$.
We sort the subtours in $W_v$ in non-decreasing order of their demands, and partition these subtours into $\frac{1}{\beta}$ groups of equal cardinality.\footnote{We add empty subtours to the first groups if needed in order to achieve equal cardinality among all groups.}
We \emph{round} the demands of the subtours in each group to the maximum demand in that group.
For each $i\in[1,m]$ and each subtour at $r_i$, the demand of that subtour is increased to the rounded value by adding a \emph{dummy} terminal of appropriate demand at vertex $r_i$.
We rearrange the subtours in $W_v$ as follows.
\begin{itemize}
\item Each subtour $t\in W_v$ in the last group is discarded, i.e., detached from the subtour in $A_1(v)$ to which it belongs.
\item Each subtour $t\in W_v$ in other groups is associated in a one-to-one manner to a subtour $t'\in W_v$ in the next group.
Letting $a$ (resp.\ $a'$) denote the subtour in $A_1(v)$ to which $t$ (resp.\ $t'$) belongs, we detach $t$ from $a$ and reattach $t$ to $a'$.
\end{itemize}
Let $A(v)$ be the set of the resulting subtours at $v$ after the rearrangement for all $t\in W_v$.
From the construction, the demand of each subtour in $A(v)$ is at most the demand of the corresponding subtour in $S^*(v)$.
Since the demand of each subtour in $A(v)$ is the sum of values in $Y$ and the set $Y$ is closed under addition (Property~3 of \cref{fact:Y}), the demand of that subtour belongs to $Y$.

\vspace{5mm}
To construct a solution $\hat S$ visiting all terminals, it remains to cover those subtours that are discarded in the construction.
For each discarded subtour $t$, we complete $t$ into a separate tour by adding the connection to the depot in both directions.
Observe that the demand of $t$ belongs to $Y$.
Let $B$ denote the set of those newly created tours.

Let $\hat S=A(r)\cup B$, where $r$ denotes the root of the tree $\hat T$.

\subsection{Analysis of $\hat S$}
It is easy to see that $\hat S$ is a feasible solution to the UCVRP, i.e., each tour in $\hat S$ starts and ends at the depot and has total demand at most 1, and each terminal is covered by some tour in $\hat S$.
The solution $\hat S$ satisfies the first four properties of the claim.
Following the analysis in Section~5.2 in~\cite{MZ22}, we obtain that $\cost(\hat S)\leq \frac{2}{2-\eps}\cdot \cost(S^*)$, where we use the bounded distance property of $T$ and the properties of $\hat T$.
Combined with the bound on $\cost(S^*)$ in \cref{thm:global}, the last property of the claim follows.

This completes the proof of the Structure Theorem (\cref{thm:structure}).

\section{Dynamic Program}
\label{sec:DP}

In this section, we prove  \cref{thm:main-hat-T}.

 \begin{theorem}
\label{thm:main-hat-T}
There is a polynomial time dynamic program that computes a solution for the UCVRP on the tree $\hat T$ with cost at most $1.5+8\eps$ times the optimal cost on the tree $T$.
\end{theorem}

\cref{thm:main} follows immediately from \cref{fact:hat-T} and \cref{thm:main-hat-T}.

\vspace{3mm}

To design the dynamic program in \cref{thm:main-hat-T}, we compute the best solution on the tree~$\hat T$ that satisfies the properties in the Structure Theorem (\cref{thm:structure}).
The algorithm consists of two phases: the first phase computes local solutions inside components (\cref{sec:DP-base}) and the second phase
computes solutions in subtrees in the bottom up order (\cref{sec:DP-subtree}).
Properties~1~and~2 of the Structure Theorem are used in \cref{sec:DP-base}; Properties~3~and~4 of the Structure Theorem are used in \cref{sec:DP-subtree}.
The analysis on the cost of the output solution is given in \cref{sec:DP-analysis}.

\subsection{Computing Solutions Inside Components}
\label{sec:DP-base}
\begin{definition}
\label{def:local-config}
A \emph{local configuration} $(c,A)$ is defined by a component $c$ and a list $A$ consisting of $\ell(A)$ pairs $(y_1,b_1), (y_2,b_2),\dots, (y_{\ell(A)},b_{\ell(A)})$ such that
\begin{itemize}
    \item $\ell(A)\leq |\mathcal{Q}_c|=O_\eps(1)$; \item for each $i\in[1,\ell(A)]$, $y_i\in Y_c$ and $b_i\in \{\text{passing, ending}\}$.
\end{itemize}
The \emph{value} of the local configuration $(c,A)$, denoted by $f(c,A)$, is the minimum cost of a collection of $\ell(A)$ subtours $t_1,\dots,t_{\ell(A)}$ in $c$ such that:
\begin{itemize}
    \item $t_1,\dots,t_{\ell(A)}$ together cover all terminals in $c$;
    \item For each $i\in [1,\ell(A)]$, the demand of $t_i$ is at most $y_i$ and the type of $t_i$ is $b_i$;\footnote{When the demand of $t_i$ is strictly less than $y_i$, a \emph{dummy} terminal will be added to $t_i$ so that the total demand of the real and the dummy terminals on $t_i$ equals $y_i$.}
    \item For each cell in $c$,    some $t_i$ visits all small terminals in that cell.
\end{itemize}
\end{definition}

When the local configuration $(c, A)$ corresponds to the solution $\hat S$ defined in the Structure Theorem (\cref{thm:structure}), that solution satisfies the above constraints using the first two properties in that Theorem.
Thus the value of that local configuration is at most the cost of $\hat S$ in $c$.

The computation for the value of a local configuration is done by exhaistive search and detailed in \cref{alg:local-config}.

\begin{algorithm}
\caption{Computation for local configurations inside a component $c$}
\label{alg:local-config}
\begin{algorithmic}[1]
\For{each list $A=\left((y_1,b_1),(y_2,b_2),\dots,(y_{\ell(A)},b_{\ell(A)})\right)$}
\State $f(c,A)\gets \infty$
\For{each partition of $\mathcal{Q}_c$ into $\ell$ parts $\mathcal{Q}_c^{(1)},\dots,\mathcal{Q}_c^{(\ell)}$ s.t.\ $\forall i$, $\demand(\mathcal{Q}_c^{(i)})\leq y_i$}
            \For{each $i\in[1,\ell]$}
                \State $U\gets\{$terminals in $\mathcal{Q}_c^{(i)}\} $
                \If{$b_i$=passing} $U\gets U\cup\{\text{the exit vertex of }c\}$\EndIf
                \State $t_i\gets$ the subtour in $c$ spanning $U$
            \EndFor
            $f(c,A)\gets \min(f(c,A),\sum_i \cost(t_i))$
\EndFor
\EndFor
\State \Return $f(c,\cdot)$
\end{algorithmic}
\end{algorithm}

\paragraph{Running time.}
Since $|\mathcal{Q}_c|=O_\eps(1)$ (\cref{fact:Y}), the number of partitions of $\mathcal{Q}_c$ is $O_\eps(1)$, and for each partition, the subtour costs can be computed in polynomial time, hence polynomial time to compute the value of a local configuration.
Observe that the number of local configurations in $c$ is $O_\eps(1)$.
Thus the overall running time to compute the values of all local configurations in $c$ is polynomial.

\subsection{Computing Solutions in Subtrees}
\label{sec:DP-subtree}
The algorithm to compute solutions in a subtree is identical to the dynamic program in Section~6.2 in~\cite{MZ22} in the equal demand setting, except for a suitable adaptation of the values of the subtour demands.
In the arbitrary demand setting, the subtour demands are values in the set $Y$ by Property~3 in the Structure Theorem (\cref{thm:structure}).
Since the set $Y$ is closed under addition (\cref{fact:Y}), the dynamic program  in~\cite{MZ22} carries over to the arbitrary demand setting (this is the place where we use Property~4 in the Structure Theorem).
Since $Y$ consists of a polynomial number of values (\cref{fact:Y}), the polynomial running time of the algorithm in \cite{MZ22} carries over to the arbitrary demand setting.

For completeness, the algorithm and the running time analysis are given in \cref{sec:DP-subtree-complete}.

\subsection{Cost Analysis}
\label{sec:DP-analysis}
The cost of the output solution is at most the cost of the solution $\hat S$ in the Structure Theorem (\cref{thm:structure}), which is at most $1.5+8\eps$ times the optimal cost for the UCVRP on the tree $T$.

This completes the proof of \cref{thm:main-hat-T}.

\appendix

\section{Hardness on the Approximation}

\subsection{UCVRP on Paths}
\label{sec:hardness-path}
We reduce the bin packing problem to the UCVRP on paths.
Consider a bin packing instance with $n$ items of sizes $a_1,\dots,a_n$, where $0< a_i\leq 1$ for each $i\in[1,n]$. We construct an instance of the UCVRP on paths as follows.
The path consists of $n+1$ vertices $r,v_1,\dots,v_n$, such that $r$ is the depot and $v_1,\dots,v_n$ are the terminals.
The weight of the edge $r v_1$ equals 1; the weight of the edge $v_i v_{i+1}$ equals 0 for each $i\in [1,n-1]$.
The demand of $v_i$ equals $a_i$ for each $i\in [1,n]$.
Observe that a solution to this instance of the UCVRP is equivalent to a solution to the bin packing instance.
Since it is NP-hard to approximate the bin packing problem to better than a 1.5 factor~\cite{williamson2011design}, it is NP-hard to approximate the UCVRP on paths to better than a 1.5 factor.

\subsection{UCVRP Inside a Component}
\label{sec:hardness-tree-component}
We reduce the bin packing problem to the UCVRP inside a component.
Consider a bin packing instance with $n$ items of sizes $a_1,\dots,a_n$, where $0< a_i\leq 1$ for each $i\in[1,n]$.
We construct a component in the UCVRP as follows.
The component consists of $n+2$ vertices $r,v_0,v_1,\dots,v_n$ which form a star at $v_0$, such that $r$ is the depot and $v_1,\dots,v_n$ are the terminals.
The weight of the edge $r v_0$ equals 1; the weight of the remaining edges $v_0 v_{i}$ equals 0 for all $i\in [1,n]$.
The demand of $v_i$ equals $a_i$ for each $i\in [1,n]$.
Observe that a solution to this instance of the UCVRP is equivalent to a solution to the bin packing instance.
Since it is NP-hard to approximate the bin packing problem to better than a 1.5 factor~\cite{williamson2011design}, it is NP-hard to approximate the UCVRP in a component to better than a 1.5 factor.

\section{Proof of \cref{lem:cluster-decomposition}}
\label{sec:cluster-decomposition}

The proof of \cref{lem:cluster-decomposition} is a minor adaptation from the proof of \cref{lem:decomposition} in \cite{MZ22}.


For any subgraph $H$ of $b$, the \emph{demand} of $H$, denoted by $\demand(H)$, is defined as the total demand of all terminals in $H$ that are \emph{strictly} inside $b$.

First, we construct \emph{leaf clusters} as follows.
For any vertex $v\in b$ such that the subtree of $b$ rooted at $v$ has demand at least $\Gamma'$ and each of the subtrees rooted at the children of $v$ has demand strictly less than $\Gamma'$, we create a leaf cluster that equals the subtree of $b$ rooted at $v$.
If the block $b$ has an exit vertex $e_b$, we need to ensure that $e_b$ is the exit vertex of some cluster. To that end, we distinguish two cases: if $e_b$  belongs to some existing leaf cluster, then we set $e_b$ to be the exit vertex of that cluster; otherwise, we create a (trivial) leaf cluster consisting of the singleton $\{e_b\}$ and set $e_b$ to be the root vertex and the exit vertex of that cluster.
See \cref{alg:cluster-leaf} for a formal description of the construction of leaf clusters.
Observe that the leaf clusters are disjoint subtrees of $b$.

\begin{definition}[key vertices]
\label{def:key-vertices}
The \emph{backbone} of $b$ is the subgraph of $b$ spanning the root of $b$ and the roots of the leaf clusters.
We say that a vertex $v\in b$ is a \emph{key vertex} if it belongs to one of the three cases: (1) $v$ is the root of the block $b$; (2) $v$ is the root of a leaf cluster; (3) $v$ has two children in  the backbone.
\end{definition}
We say that two key vertices  $v_1$ and $v_2$ are \emph{consecutive} if the $v_1$-to-$v_2$ path in the tree does not contain any other key vertex.
For each pair of consecutive key vertices $(v_1,v_2)$, we consider the subgraph \emph{between} $v_1$ and $v_2$, and decompose that subgraph into \emph{internal clusters}, each of demand at most $2\Gamma$, such that all of these cluster are \emph{big} (i.e., of demand at least $\Gamma$) except possibly for the upmost cluster.
See \cref{alg:cluster-internal} for the detailed construction of internal clusters.

\begin{algorithm}[t]
\caption{Construction of leaf clusters in block $b$.}
\label{alg:cluster-leaf}
\begin{algorithmic}[1]
\For{each vertex $v\in b$}
    \State $T(v)\gets$ subtree of $b$ rooted at $v$
\EndFor
\For{each non-leaf vertex $v\in b$}
    \State Let $v_1$ and $v_2$ denote the two children of $v$ in $b$
    \If{$\demand(T(v))\geq \Gamma'$ and $\demand(T(v_1))<\Gamma'$ and $\demand(T(v_2))<\Gamma'$}
        \State Create a leaf cluster that equals $T(v)$ and has root vertex $v$
    \EndIf
\EndFor
\If{$b$ has an exit vertex $e_b$}
    \If{$e_b$ belongs to some existing leaf cluster $x$}
        \State $e_x\gets e_b$
    \Else
        \State Create a leaf cluster $x^*$ that equals a singleton $\{e_b\}$
        \State $e_{x^*}\gets e_b$
        \State $r_{x^*}\gets e_b$
    \EndIf
\EndIf
\end{algorithmic}
\end{algorithm}

\begin{algorithm}[t]
\caption{Construction of internal clusters in block $b$.}
\label{alg:cluster-internal}
\begin{algorithmic}[1]
\For{each pair of \emph{consecutive} key vertices $v_1, v_2$ such that $v_1$ is an ancester of $v_2$}
    \State $v_1'\gets$ child of $v_1$ on the $v_1$-to-$v_2$ path
    \For{each vertex $v$ on the $v_1'$-to-$v_2$ path}
        $H(v)\gets T(v)\setminus T(v_2)$
    \EndFor
    \State $x\gets v_2$
    \While{$H(v_1')$ has demand at least $\Gamma'$}
        \State $v\gets$ the deepest vertex on the $v_1'$-to-$x$ path such that $H(v)$ has demand at least $\Gamma'$
        \State Create an internal cluster that equals $H(v)$ with root vertex $v$ and exit vertex $x$
        \State $x\gets v$
        \For{each vertex $v'$ on the $v_1'$-to-$x$ path}
             $H(v')\gets T(v')\setminus T(x)$
        \EndFor
    \EndWhile
    \State Create an internal cluster that equals $\{(v_1,v_1')\}\cup H(v_1)$ with root vertex $v_1$ and exit vertex~$x$
\EndFor
\end{algorithmic}
\end{algorithm}

The first three properties in the claim follow from the construction.
It remains to show the last property in the claim. The following lemma is crucial in the analysis.
\begin{lemma}
\label{lem:three-pre-images}
We say that a cluster in $b$ is \emph{good} if it is a leaf cluster or a big cluster, and is \emph{bad} otherwise.
There is a map from all clusters in $b$ to good clusters in $b$ such that each good cluster in $b$ has at most three pre-images.
\end{lemma}
\begin{proof}
First, we define a map $\phi_{\rm good}$ that maps each good cluster to itself.
Next, we consider the bad clusters.
Observe that the root vertex $r_x$ of any bad cluster $x$ is a key vertex.
We say that a bad cluster $x$ is a \emph{left bad} cluster (resp.\ \emph{right bad} cluster) if $x$ contains the left child (resp.\ right child) of $r_x$.
We define a map $\phi_{\rm left}$ from left bad clusters to leaf clusters, such that the image of a left bad cluster is the leaf cluster that is \emph{rightmost} among its descendants.
We show that $\phi_{\rm left}$ is injective.
Let $c_1$ and $c_2$ be any left bad clusters.
Observe that $r_{c_1}$ and $r_{c_2}$ are distinct key vertices.
If $r_{c_1}$ is ancestor of $r_{c_2}$ (the case when $r_{c_2}$ is ancestor of $r_{c_1}$ is similar), then $\phi_{\rm left}(c_2)$ is in the left subtree of $r_{c_2}$ whereas $\phi_{\rm left}(c_1)$ is outside the left subtree of $r_{c_2}$, so $\phi_{\rm left}(c_1)\neq \phi_{\rm left}(c_2)$.
In the remaining case, the subtrees rooted at $r_{c_1}$ and at $r_{c_2}$ are disjoint, so $\phi_{\rm left}(c_1)\neq \phi_{\rm left}(c_2)$.
Thus $\phi_{\rm left}$ is injective.
Note that every leaf cluster is good, so $\phi_{\rm left}$ is an injective map from left bad clusters to good clusters.
Similarly, we obtain an injective map $\phi_{\rm right}$ from right bad clusters to good clusters.
The claim follows by combining the three injective maps $\phi_{\rm good}$, $\phi_{\rm left}$, and $\phi_{\rm right}$.
\end{proof}

Observe that all good clusters are big except possibly one cluster, i.e., the cluster $x^*$ containing the exit vertex $e_b$ of the block.
Thus the number of good clusters is at most $\demand(b)/\Gamma'+1$.
By \cref{lem:three-pre-images}, the number of clusters in $b$ is at most three times the number of good clusters in $b$, hence the last property of the claim.
This completes the proof of \cref{lem:cluster-decomposition}.

\section{Computing Solutions in Subtrees}
\label{sec:DP-subtree-complete}

This section is a slight adaptation from \cite{MZ22}.
Everything is identical to \cite{MZ22} except that we use the properties on the set $Y$ for the subtour demands (\cref{fact:Y}).

\begin{definition}
\label{def:subtree-config}
A \emph{subtree configuration} $(v,A)$ is defined by a vertex $v$ and a list $A$ consisting of $\ell(A)$ pairs $(\tilde y_1,n_1)$, $(\tilde y_2,n_2),\ldots,(\tilde y_\ell,n_{\ell(A)})$ such that
\begin{itemize}
\item $v$ belongs to the set $I$ (\cref{def:I});
\item
$\ell(A)=O_\eps(1)$; in particular, when $v$ is a critical vertex, $\ell(A)\leq \left(\frac{1}{\beta}\right)^{\frac{1}{\alpha}}$;
\item
for each $i\in[1, \ell(A)]$,  $\tilde y_i$ belongs to the set $Y$ (\cref{def:Y}) and $n_i$ is an integer in $[0,n]$.
\end{itemize}
\end{definition}
The \emph{value} of the subtree configuration $(v,A)$, denoted by $g(v,A)$, is the minimum cost of a collection of $\ell(A)$ subtours in the subtree of $\hat T$ rooted at $v$, each subtour starting and ending at $v$, that together visit all of the real terminals of the subtree rooted at $v$, such that, for each $i\in[1,\ell(A)]$, there are $n_i$ subtours each with total demand of the real and the dummy terminals that equals $\tilde y_i$.

When $v$ is the root $r$, any subtree configuration $(r,A)$ corresponds to a solution to the UCVRP on the tree $\hat T$.
The output of the algorithm is the minimum value of $g(r,A)$ among all lists $A$.

To compute the values of subtree configurations, we consider the vertices $v\in I$ in the bottom up order.
For each vertex $v\in I$ that is the root of a component, we compute the values $g(v,\cdot)$ using \cref{alg:subtree-configuration-root} in \cref{sec:subtree-configuration-root}; and for each vertex $v\in I$ that is a critical vertex, we compute the values $g(v,\cdot)$ using \cref{alg:subtree-configuration-critical} in \cref{sec:subtree-configuration-critical}.
See \cref{fig:DP}.

\begin{figure}[t]
    \centering
    \includegraphics[scale=0.3]{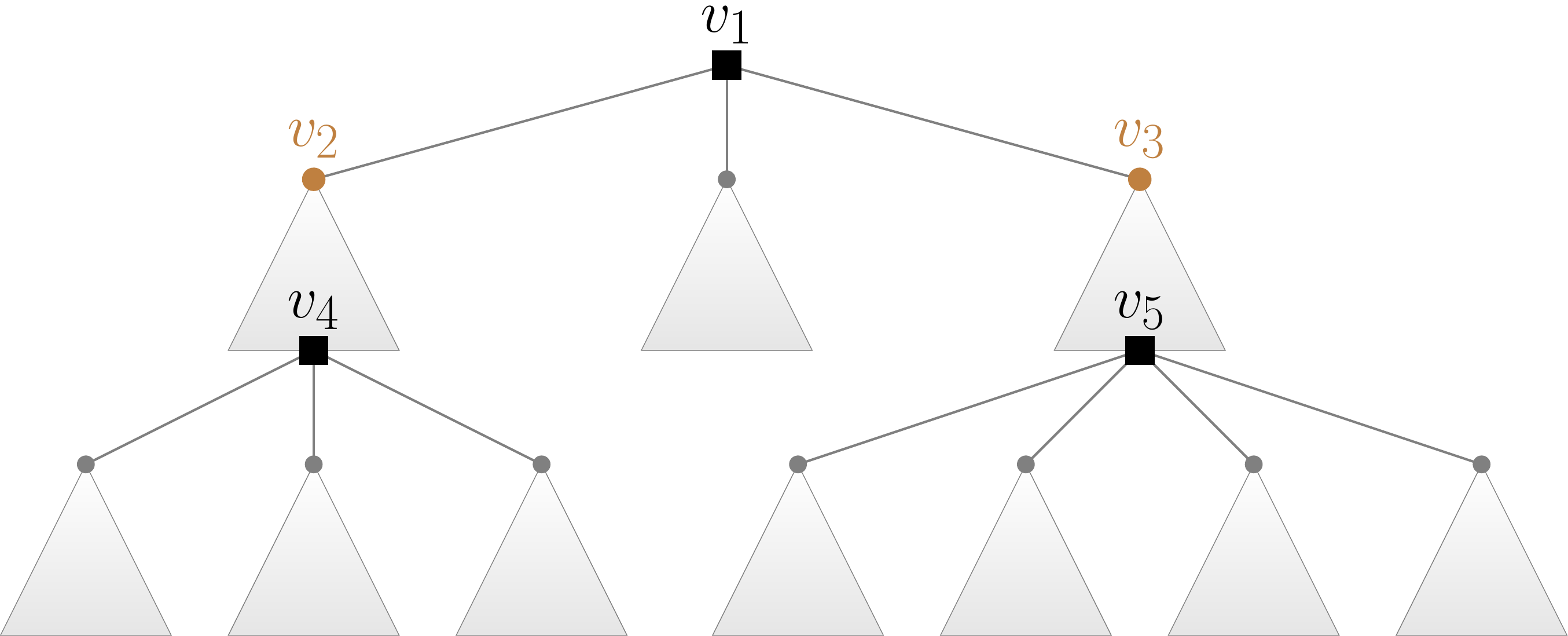}
    \caption{\small At each vertex of the tree, the dynamic program memorizes the demands of the subtours in the subtree and their total cost.
    Here is an example of the flow of execution in the dynamic program.
    First, independent computation in each component (\cref{alg:local-config}). Next, computation in the subtrees rooted at vertices $v_4$ and $v_5$ (\cref{alg:subtree-configuration-critical}).
    Then computation in the subtrees rooted at vertices $v_2$ and $v_3$ (\cref{alg:subtree-configuration-root}).
    Finally, computation in the subtree rooted at vertex $v_1$ (\cref{alg:subtree-configuration-critical}).
    The output is the best solution in the subtree
    rooted at $v_1$.
    }
    \label{fig:DP}
\end{figure}

\subsection{Subtree Configurations at the Root of a Component}
\label{sec:subtree-configuration-root}
In this subsection, we compute the values of the subtree configurations at the root $r_c$ of a component~$c$.

From \cref{sec:local}, we have already computed the values of the local configurations in the component $c$.
If $c$ is a leaf component, the local configurations in $c$ induce the subtree configurations at $r_c$, where $\tilde y_i$ is the demand of the $i$-th subtour in the local configuration and $n_i=1$ for each $i$.
In the following, we consider the case when $c$ is an internal component.
From \cref{def:critical}, we observe that the exit vertex $e_c$ of the component $c$ is a critical vertex.
Thus the values of subtree configurations at $e_c$ have already been computed using \cref{alg:subtree-configuration-critical} in \cref{sec:subtree-configuration-critical} according to the bottom up order of the computation.
To compute the value of a subtree configuration at $r_c$, we combine a subtree configuration at $e_c$ and a local configuration in $c$, in the following way.

Consider a subtree configuration $(e_c,A_e)$ and a local configuration $(c,A_c)$, where
\[A_e=((\tilde y_1,n_1),(\tilde y_2,n_2),\ldots,(\tilde y_{\ell_e},n_{\ell_e})),\]
\[A_c=((y_1,b_1),(y_2,b_2),\dots,(y_{\ell_c},b_{\ell_c})).\]
To each $i\in[1,\ell_c]$ such that $b_i$ is ``passing'', we associate $y_i$ with $\tilde y_j$ for some $j\in[1,\ell_e]$ with the constraints that $y_i+\tilde y_j\leq 1$ and for each $j\in[1,\ell_e]$ at most $n_j$ elements are associated to $\tilde y_j$ (because in the subtree rooted at $e_c$ we only have $n_j$ subtours of demand $\tilde y_j$ at our disposal).
Observe that, for each association $(y_i,\tilde y_j)$, we have $y_i\in Y_c\subseteq Y$ (\cref{fact:Y}), $\tilde y_j\in Y$, and $y_i+\tilde y_j\leq 1$.
Thus $y_i+\tilde y_j\in Y$ by \cref{fact:Y}.
Consequently, we obtain the list $A$ of a subtree configuration $(r_c,A)$ as follows:
\begin{itemize}
    \item For each association $(y_i,\tilde y_j)$, we put in $A$ the pair $(y_i+\tilde y_j,1)$.
    \item For each pair $(\tilde y_j,n_j)\in A_e$, we put in $A$ the pair $(\tilde y_j,n_j-(\text{number of $y_i$'s associated to $\tilde y_j$}))$.
    \item For each pair $(y_i,``\text{ending}")\in A_c$, we put in $A$ the pair $(y_i,1)$.
\end{itemize}
From the construction, $\ell(A)\leq \ell(A_e)+\ell(A_c)$.
Since $e_c$ is a critical vertex, $\ell(A_e)\leq \left(\frac{1}{\beta}\right)^{\frac{1}{\alpha}}$ by \cref{def:subtree-config}.
From \cref{def:local-config}, $\ell(A_c)=O_\eps(1)$.
Thus $\ell(A)\leq \left(\frac{1}{\beta}\right)^{\frac{1}{\alpha}}+O_\eps(1)=O_\eps(1)$.

Next, we compute the cost of the combination of the subtree configuration $(e_c,A_e)$ and the local configuration $(c,A_c)$; let $x$ denote this cost.
For any subtour $t$ at $e_c$ that is not associated to any non-spine passing subtour in the component $c$, we pay an extra cost to include the spine subtour of the component $c$, which is combined with the subtour $t$.
The number of times that we include the spine subtour of $c$ is the number of subtours at $e_c$ minus the number of passing subtours in $A_c$, which is $\sum_{j\leq \ell_e} n_j- \sum_{i\leq \ell_c} \mathbbm{1} \left[b_i \text{ is ``passing''}\right]$.
Thus we have
\begin{equation}  \label{eq:costcomputation}
x=f(c,A_c)+g(e_c,A_e)+\cost(\spine_c)\cdot\left(\bigg(\sum_{j\leq \ell_e} n_j\bigg) - \bigg(\sum_{i\leq \ell_c} \mathbbm{1} \left[b_i \text{ is ``passing''}\right]\bigg)\right).\end{equation}

The algorithm is described in \cref{alg:subtree-configuration-root}.

\begin{algorithm}[t]
\caption{Computation for subtree configurations at the root of a component $c$.}
\label{alg:subtree-configuration-root}
\begin{algorithmic}[1]
\For{each list $A$}
    \State $g(r_c,A)=+\infty$.
\EndFor
\For{each subtree configuration $(e_c,A_e)$ and each local configuration $(c,A_c)$}
        \For{each way to combine $(e_c,A_e)$ and $(c,A_c)$}
            \State $A\gets$ the resulting list
            \State $x\gets$ the cost computed in Equation~(\ref{eq:costcomputation})
            \State $g(r_c,A)\gets \min (g(r_c,A),x)$.
        \EndFor
\EndFor
\State \Return $g(r_c,\cdot)$
\end{algorithmic}
\end{algorithm}

\paragraph*{Running time}
Since $|Y|=n^{O_\eps(1)}$ (\cref{fact:Y}), the number of subtree configurations $(e_c,A_e)$ and the number of local configurations $(c,A_c)$ are both $n^{O_\eps(1)}$.
For fixed $(e_c,A_e)$ and $(c,A_c)$, the number of ways to combine them is $O_\eps(1)$.
Thus the running time of \cref{alg:subtree-configuration-root} is $n^{O_\eps(1)}$.

\subsection{Subtree Configurations at a Critical Vertex}

\label{sec:subtree-configuration-critical}
In this subsection, we compute the values of the subtree configurations at a critical vertex.
Let $z$ denote any critical vertex.

Consider the solution $\hat S$ in the Structure Theorem (\cref{thm:structure}).
By Property~4 in the Structure Theorem, there exists a set $X\subseteq Y$ of $\frac{1}{\beta}$ values such that the demand of each subtour at a child of $z$ is among the values in $X$.
Since $Y\subseteq [\alpha,1]$, we have $X\subseteq [\alpha,1]$.
Thus the demand of a subtour at $z$ is the sum of at most $\frac{1}{\alpha}$ values in $X$.
Therefore, the number of distinct demands of the subtours at $z$ is at most $(\frac{1}{\beta})^{\frac{1}{\alpha}}=O_\eps(1)$.
In addition, those demands belong to $Y$, since each demand is the sum of a multiset of values in $X$ and using \cref{fact:Y}.

To compute a solution satisfying Property~4 in the Structure Theorem, we enumerate all sets $X\subseteq Y$ of $\frac{1}{\beta}$ values, compute a solution with respect to each set $X$, and return the best solution found.
Unless explicitly mentioned, we assume in the following that the set $X$ is fixed.

\begin{definition}[sum list]
A \emph{sum list} $A$ consists of $\ell(A)$ pairs
$(y_1,n_1)$, $(y_2,n_2)$, \dots, $(y_{\ell(A)},n_{\ell(A)})$ such that
\begin{enumerate}
    \item
$\ell(A)\leq (\frac{1}{\beta})^{\frac{1}{\alpha}}$;
\item
For each $i\in[1,\ell(A)]$, $y_i\in Y$ is the sum of a multiset of values in $X$ and $n_i$ is an integers in $[0,n]$.
\end{enumerate}
\end{definition}
From the Structure Theorem, we only need to consider subtree configurations $(z,A)$ such that the list $A$ is a sum list.

Let $r_1,r_2,\ldots ,r_m$ be the children of $z$.
For each $i\in[1,m]$, let \[A_i=( (y_1^{(i)},n_1^{(i)}),(y_2^{(i)},n_2^{(i)}),\ldots , (y_{\ell_i}^{(i)},n_{\ell_i}^{(i)}))\] denote the list in a subtree configuration $(r_i,A_i)$.
We \emph{round} the list $A_i$ to a list \[\overline{A_i}=( (\overline{y_1^{(i)}},n_1^{(i)}),\overline{(y_2^{(i)}},n_2^{(i)}),\ldots , (\overline{y_{\ell_i}^{(i)}},n_{\ell_i}^{(i)})),\] where $\overline{x}$  denotes the smallest value in $X$ that is greater than or equal to $x$, for any value $x$.
The rounding is represented by adding a dummy terminal of demand $\overline{x}-x$ at vertex $r_i$ to each subtour of initial demand $x$.

Let $\mathcal{S}\subseteq Y$ denote a multiset such that for each $i\in[1,m]$ and for each $j\in[1,\ell_i]$, the multiset $\mathcal{S}$ contains $n_j^{(i)}$ copies of $\overline{y_j^{(i)}}$.

\begin{definition}[compatibility]
A multiset $\mathcal{S}\subseteq Y$ and a sum list \[A=( (y_1,n_1),(y_2,n_2),\ldots , (y_{\ell(A)},n_{\ell(A)}))\] are \emph{compatible} if there is a partition of $\mathcal{S}$ into $\sum_{i=1}^{\ell(A)} n_i$ parts and a correspondence between the parts of the partition and the values $y_i$'s, such that for each $y_i$, there are $n_i$ associated parts, and for each of those parts, the elements in that part sum up to $y_i$.
\end{definition}

For a sum list $A=( (y_1,n_1),(y_2,n_2),\ldots , (y_{\ell(A)},n_{\ell(A)}))$, the value $g(z,A)$ of the subtree configuration $(z,A)$ equals the minimum, over all sets $X$ and all subtree configurations $\{(r_i,A_i)\}_{1\leq i\leq m}$ such that $\mathcal{S}$ and $A$ are compatible, of
\begin{equation}
\label{eqn:g-z-A}
\sum_{i=1}^m g(r_i,A_i)+ 2\cdot n(A_i)\cdot w(r_i,z),
\end{equation}
where $n(A_i)$ denotes $\sum_{j=1}^{\ell_i} n_j^{(i)}$.
We note that $n_1 y_1+n_2 y_2+\cdots +n_{\ell(A)} y_{\ell(A)}$ is equal to the total demand of the (real and dummy) terminals in the subtree rooted at $z$.

\begin{algorithm}[t]
\caption{Computation for subtree configurations at a critical vertex $z$.}
\label{alg:subtree-configuration-critical}
\begin{algorithmic}[1]
    \For{each list $A$}
        \State $g(z,A)\gets+\infty$
    \EndFor
    \For{each set $X\subseteq Y$ of $\frac{1}{\beta}$ values}
        \For{each $i\in [0,m]$ and each list $A$}
            \State $\DP_i(A)\gets +\infty$
        \EndFor
        \State $\DP_0(\emptyset)\gets 0$
        \For{each $i\in[1,m]$}
            \For{each subtree configuration $(r_i,A_i)$}
                \State $\overline{A_i}\gets round(A_i)$
                \For{each sum list $A_{\leq i-1}$}
                    \For{each way to combine $A_{\leq i-1}$ and $\overline{A_i}$}
                        \State $A_{\leq i}\gets$ the resulting sum list
                        \State $x\gets \DP_{i-1}(A_{\leq i-1})+g(r_i,A_i)+2\cdot n(A_i)\cdot w(r_i,z)$
                        \State $\DP_i(A_{\leq i})\gets \min(\DP_i(A_{\leq i}),x)$
                    \EndFor
                \EndFor
            \EndFor
    \EndFor
    \For{each list $A$}
        \State $g(z,A)\gets \min(g(z,A),\DP_m(A))$
    \EndFor
\EndFor
\State \Return $g(z,\cdot)$
\end{algorithmic}
\end{algorithm}

Fix any set $X\subseteq Y$ of $\frac{1}{\beta}$ values.
We show how to compute the minimum cost of~\cref{eqn:g-z-A} over all subtree configurations $\{(r_i,A_i)\}_{1\leq i\leq m}$ such that $\mathcal{S}$ and $A$ are compatible.
For each $i\in[1,m]$ and for each subtree configuration $(r_i,A_i)$, the value $g(r_i,A_i)$ has already been computed using \cref{alg:subtree-configuration-root} in \cref{sec:subtree-configuration-root}, according to the bottom up order of the computation.
We use a dynamic program that scans $r_1,\dots, r_m$ one by one: those are all siblings, so here the reasoning is not bottom-up but left-right.
Fix any $i\in[1,m]$.
Let $\mathcal{S}_i\subseteq Y$ denote a multiset such that for each $i'\in[1,i]$ and for each $j\in[1,\ell_{i'}]$, the multiset $\mathcal{S}_i$ contains $n_j^{(i')}$ copies of $\overline{y_j^{(i')}}$.
We define a dynamic program table $\DP_i$.
The value $\DP_i(A_{\leq i})$ at a sum list $A_{\leq i}$ equals the minimum, over all subtree configurations $\{(r_{i'},A_{i'})\}_{1\leq i'\leq i}$ such that $\mathcal{S}_i$ and $A_{\leq i}$ are compatible, of   \[\sum_{i'=1}^i g(r_{i'},A_{i'})+ 2\cdot n(A_{i'})\cdot w(r_{i'},z).\]
When $i=m$, the values $\DP_m(\cdot)$ are those that we are looking for.
It suffices to fill in the tables $\DP_1,\dots,\DP_m$.

To compute the value $\DP_i$ at a sum list $A_{\leq i}$, we use the value $\DP_{i-1}$ at a sum list $A_{\leq i-1}$ and the value $g(r_i,A_i)$ of a subtree configuration $(r_i,A_i)$.
Let $A_{\leq i-1}=( (\hat y_1,\hat n_1),(\hat y_2,\hat n_2),\ldots , (\hat y_{\ell},\hat n_{\ell}))$.
We combine $A_{\leq i-1}$ and $\overline{A_i}$ as follows.
For each $p\in [1,\ell]$ and each $j\in[1,\ell_i]$ such that $\hat y_p+\overline{y_{j}^{(i)}}\leq 1$, we observe that $\hat y_p+\overline{y_{j}^{(i)}}$ is the sum of a multiset of values in $X$, thus belongs to $Y$ by \cref{fact:Y}.
We create $n_{p,j}$ copies of the association of $(\hat y_p,\overline{y_{j}^{(i)}})$, where $n_{p,j}\in [0,n]$ is an integer variable that we enumerate in the algorithm.
We require that for each $p\in[1,\ell]$, $\sum_{j=1}^{\ell_i} n_{p,j}\leq \hat n_p$; and for each $j\in[1,\ell_i]$, $\sum_{p=1}^{\ell} n_{p,j} \leq n_j^{(i)}$.
The resulting sum list $A_{\leq i}$ is obtained as follows.
\begin{itemize}
\item For each association $(\hat y_p,\overline{y_{j}^{(i)}})$, we put in $A_{\leq i}$ the pair $(\hat y_p+\overline{y_{j}^{(i)}},n_{p,j})$.
\item For each pair $(\hat y_p,\hat n_p)\in A_{\leq i-1}$, we put in $A_{\leq i}$ the pair $(\hat y_p,\hat n_p-\sum_{j=1}^{\ell_i} n_{p,j})$.
\item For each pair $(\overline{y_j^{(i)}},n_j^{(i)})\in \overline{A_i}$, we put in $A_{\leq i}$ the pair $(\overline{y_j^{(i)}},n_j^{(i)}-\sum_{p=1}^{\ell} n_{p,j})$.
\end{itemize}

The algorithm is described in \cref{alg:subtree-configuration-critical}.

\paragraph{Running time}
Since $|Y|=n^{O_\eps(1)}$, the numbers of subtree configurations and of sum lists are $n^{O_\eps(1)}$.
Observe that the number of ways to combine them and the number of sets $X$ are $n^{O_\eps(1)}$.
Thus the overall running time of \cref{alg:subtree-configuration-critical} is $n^{O_\eps(1)}$.

\bibliographystyle{alpha}
\bibliography{references}

\end{document}